\documentclass{lmcs}
\usepackage{lmodern}
\usepackage{tikz}
\usetikzlibrary{intersections, calc, arrows, automata, positioning,shapes, trees}
\usetikzlibrary{bending, arrows.meta}
\usetikzlibrary{decorations.pathmorphing}
\usepackage{bbding}
\usepackage{pifont}
\usepackage{pgf, verbatim}
\usepackage{amssymb}
\usepackage{amsmath,amsthm}
\usepackage{comment}
\usepackage{xspace}
\usepackage{booktabs}
\usepackage{diagbox}
\usepackage{makecell}
\usepackage{hyperref}
\usepackage[capitalise]{cleveref}
\newcommand{\game}{{\mathcal{G}}}
\newcommand{\player}[1]{\mathsf{Player}(#1)}

\newcommand{\Ad}{{\mathsf{Adam}}}
\newcommand{\Aa}{{\mathcal{A}}}
\newcommand{\Ev}{{\mathsf{Eve}}}
\newcommand{\Gg}{{\mathcal G}} 
 
\newcommand{\nat}{{\mathbb N}} %
\newcommand{\F}{{\mathcal{F}}}

\newcommand{\sg}{\sigma}

\newcommand{\AtE}{\mathsf{Attr}_{\Ev}}

\newcommand{\genreach}{\mathsf{GenReach}}
\newcommand{\maxgenreach}{\mathsf{MaxGenReach}}

\newcommand{\maxgenreachpromise}{\mathsf{MaxGenReachPromise}}

\newcommand\PSPACE{$\textsc{PSPACE}$\xspace}
\newcommand\NPSPACE{$\textsc{NPSPACE}$\xspace}
\newcommand\APTime{$\textsc{APTime}$\xspace}

\newcommand\NL{$\textsc{NL}$\xspace}
\newcommand\NP{$\textsc{NP}$\xspace}
\newcommand\coNL{$\textsc{co}\textsc{NL}$\xspace}
\newcommand\coNP{$\textsc{co}\textsc{NP}$\xspace}
\renewcommand\P{$\textsc{PTIME}$\xspace}

\newcommand{\M}{\mathcal{M}}
\newcommand{\upd}{\mu}
\newcommand{\ksurdeux}{\lfloor k/2 \rfloor}
\newcommand{\binomksurdeuxk}{\binom{k}{\ksurdeux}}
\newcommand{\set}[1]{\{ #1 \}}

\usepackage{tcolorbox}

\usepackage{lineno}
\usepackage{todonotes}
\usepackage{iftex}

\ifpdf
  \usepackage{underscore}         %
  \usepackage[T1]{fontenc}        %
\else
  \usepackage{breakurl}           %
\fi

\begin{document}

\title{Generalised Reachability Games}

% affiliations are numbered automatically with a, b, c (see below)
% use the optional argument to indicate the affiliation(s) of each author
% omit the argument if there is only one author, or only one affiliation
\author[S.~Bose]{Sougata Bose\lmcsorcid{0000-0003-3662-3915}}[a]
\author[N.~Fijalkow]{Nathana\"el Fijalkow\lmcsorcid{0000-0002-6576-4680}}[b]
\author[D.~Hausmann]{Daniel Hausmann\lmcsorcid{0000-0002-0935-8602}}[c]
\author[F.~Horn]{Florian Horn\lmcsorcid{0000-0001-8872-4705}}[d]
\author[S.~Paul]{Soumyajit Paul\lmcsorcid{0000-0002-7233-2018}}[c]
\author[S.~Schewe]{Sven Schewe\lmcsorcid{0000-0002-9093-9518}}[c]
\author[T.~Zhanabekova]{Tansholpan Zhanabekova\lmcsorcid{0000-0002-4941-2554}}[c]

% affiliation 1 (automatically numbered a)
\address{UMONS - Universit\'e de Mons, Mons, Belgium}	%optional
% write emails for all authors having that affiliation
%\email{name1@email1}  %optional

% affiliation 2 (automatically numbered a)
\address{CNRS, LaBRI, Bordeaux, France}	%optional
% write emails for all authors having that affiliation
%\email{name1@email2}  %optional

% affiliation 3 (automatically numbered c)
\address{University of Liverpool, Liverpool, UK}	%optional
%\email{name1@email3, name3@email3, name4@email4}  %optional

% affiliation 3 (automatically numbered c)
\address{Université Paris Cité, CNRS, IRIF, F-75013 Paris, France}	%optional
%\email{name2@email2}  %optional

\begin{abstract}
We study two-player zero-sum turn-based games played on graphs with multiple reachability objectives called \emph{generalised reachability games}.
In classic reachability games the goal of one player -- $\Ev$, is to visit a given target set of vertices, and that of the other player -- $\Ad$, is to prevent this. In generalised reachability games, 
%studied by Fijalkow and Horn, 
the single target set is replaced with a family of target sets and the objective of $\Ev$ is to visit all of them in any order. We study the complexity of deciding the winner in two-player games with generalised reachability objectives.

Our study reveals that an important parameter that determines the complexity of this problem is the size of the target sets. We first prove that deciding the winner in such games is \textsc{PSPACE}-complete, and the \textsc{PSPACE} lower bound holds even when the size of each target set is at most three. By contrast, we show that the problem is FPT in the number of target sets of size greater than one. Moreover, we consider the memory requirements for both players and give matching upper and lower bounds
on the sizes of winning strategies.

%Our results are twofold: first, we provide an improved complexity picture for generalised reachability games, expanding the known tractable class from games in which all target sets are singleton to additionally allowing a logarithmic number of target sets of arbitrary size. 

We also study optimisation variants of these games. For the optimisation problems, we show intractability for most interesting cases. Particularly, in contrast to the tractability of generalised reachability in the case with singleton target sets, the optimisation problem is co\NP-hard when $\Ev$ tries to maximise the number of target sets that are visited. Tractability of this case can be recovered in a different optimisation setting where $\Ev$ is required to pledge a maximum sized subset of target sets that she can guarantee to visit. 
\end{abstract}

\maketitle

\section{Introduction}\label{sec:intro}
Two-player zero-sum games played on graphs provide a fundamental framework for modelling decision-making scenarios where two players have opposing objectives. They are extensively used to model reactive systems, where one player ($\Ev$) models actions controlled by the system and the other player ($\Ad$) models uncontrollable actions of the environment. Analysing such games then provides formal guarantees on the behaviours of the system against all possible behaviours of the environment~\cite{fijalkow2023gamesgraphs,pnueli1977, PR89, Bloem2018}.

These games are played on finite graphs, where the vertices are partitioned into ones controlled by $\Ev$ and $\Ad$. Starting from a token placed on a fixed initial vertex, the player controlling the current vertex moves the token along an edge of the graph to jointly form an infinite path. A winning objective specifies the set of acceptable behaviours of the system as a set of infinite paths that are good for $\Ev$. In the game, $\Ev$ attempts to ensure that the path formed is good for $\Ev$, while $\Ad$ tries to obstruct this goal.
Solving games refers to the decision problem of checking which of the two players can win from a given initial vertex.

A key type of objective studied in such settings is \emph{reachability}, where the goal of $\Ev$ is to reach some vertex among a designated target set of vertices. 
Generalised reachability games extend this concept by requiring $\Ev$ to reach multiple target sets rather than just one. Then, $\Ev$ is required to visit at least one vertex from each target set in the play.
However, in terms of computational complexity of solving such games, changing from reachability to generalised reachability results in a significant jump. While reachability games are \P-complete \cite{Zer13,DBLP:journals/jcss/Immerman81}, generalised reachability games are \PSPACE-complete (Theorem~\ref{thm:pspace-hardness-general-case}). 
Understanding the complexity of solving special cases of generalised reachability games is crucial, especially when considering the size of target sets as a parameter. %For instance, if all target sets consist of a single vertex, $\Ev$ is required to visit several vertices in the play. Such games can be solved in polynomial time (Theorem~\ref{thm:ptime-singleton-2player}, first proven in~\cite{fij-horn2010}.

The analysis of generalised winning conditions, i.e., the conjunction of multiple objectives of the same kind has been a significant focus of research~\cite{chatterjee2007, chatterjee_et_al10, chatterjee_et_al16}.  
While generalised reachability games have been well studied, one could also consider the optimisation variant of the problem that asks $\Ev$
to visit as many target sets as possible.
This variant can be used to provide strategies where visiting all target sets may not be possible, but $\Ev$ can still visit a significant number of target sets.
Optimisation variants have been studied beyond generalised reachability objectives~\cite{KupfermanS24}.
In this work, we consider both the case where $\Ev$ has to name the target sets she visits before the game starts (and visiting other sets does not count) \emph{and} the case where she just wants to maximise the number of sets visited.

\paragraph{Our Contributions.}
The current work provides a comprehensive study of generalised reachability games, joining and extending~\cite{DBLP:journals/corr/abs-2509-14091} and~\cite{fij-horn2010,FijalkowHorn13french}. 

We first analyse the complexity of deciding their winner, showing that the size of the target sets plays a crucial role. We then study the memory requirements of winning strategies for both players. Finally, we consider natural optimisation variants of the problem with relaxed winning objectives and analyse their complexity. In particular, we show the following results.

\begin{itemize}
    \item[--] We prove that deciding the winner in generalised reachability games is \PSPACE-complete in general. It is \PSPACE-hard even when the sizes of the target sets are at most $3$ (Theorem~\ref{thm:pspace-hardness-general-case}). On the other hand, we show that the problem is fixed-parameter tractable when parameterised by the number of target sets with size bigger than $1$ (Theorem~\ref{thm:fpt-2-player}).
    
    \item[--] We study the complexity of one-player restrictions of the problem. When all vertices belong to $\Ev$, deciding if $\Ev$ can win is \NP-complete (Theorem~\ref{thm:genreach-eve-only}). When all vertices belong to $\Ad$, the same problem is \NL-complete (Theorem~\ref{thm:genreach-adam-only}).
    
    \item[--] We analyse the memory requirements for both players and give upper and lower bounds on the size of winning strategies. In a game with $k$ target sets, we show that $\Ev$ may need $2^k - 1$ memory states, while $\Ad$ may need $\binom{k}{\lfloor k/2 \rfloor}$ memory states (Theorem~\ref{thm:memlower}). We show that these bounds are tight (Theorem~\ref{thm:memupper}). 
    \item[--] We analyse two optimisation variants of the problem. In the first variant $\Ev$ maximises the number of targets she can reach. In the second variant $\Ev$ maximises the same objective but under the constraint that she has to announce which targets she will visit (in any order) before the game is played. We show that the first variant is \PSPACE-complete, with a \PSPACE lower bound holding even for the case where all target sets have size at most $2$ (Theorem~\ref{thm:maxgrcomplexity}.(1)). For the case that all target sets have size $1$, we show that this problem is \coNP-hard (Theorem~\ref{thm:maxgrcomplexity}). On the other hand, the complexity of the second variant behaves similarly to standard generalised reachability games. We show that it is \PSPACE-complete in general and can be solved in polynomial time when all target sets have size $1$ (Theorem~\ref{thm:maxgrp2p}). We also provide complexity bounds for one-player restrictions of these problems. 
\end{itemize}

The resulting complexity landscapes of all decision problems are summarised in the tables in Figures~\ref{fig:complexity-table},~\ref{fig:maxgenreach-complexity-table}, and~\ref{fig:maxgenreachpromise-complexity-table}.
\begin{figure}
  \centering
	\scalebox{.95}{
	\renewcommand{\arraystretch}{2.2}
	\setlength{\tabcolsep}{8pt}
	\begin{tabular}{|l|c|c|c|c|c| }
	\hline
	\small \diagbox{Players}{Targets} & \small All size $1$ & \small $k$ size $> 1$ & \small All size $\leq 2$ & \small All size $\leq 3$ & \small General \\
    \hline
	2 Players & \makecell{\P \\ \small (Thm~\ref{thm:ptime-singleton-2player})} & \makecell{\text{FPT} \\ \small (Thm~\ref{thm:fpt-2-player})} & {open} & \multicolumn{2}{|c|}{\makecell{\PSPACE-complete \\ \small (Thm~\ref{thm:pspace-hardness-general-case})}}\\
	\hline
	Only Eve & \makecell{\NL-complete \\ \small (Thm~\ref{thm:genreach-eve-only})} & \multicolumn{2}{|c|}{\makecell{\P \\ \small (Thm~\ref{thm:genreach-eve-only})}} & \multicolumn{2}{|c|}{\makecell{\NP-complete \\ \small (Thm~\ref{thm:genreach-eve-only})}}\\
	\hline
	Only Adam & \multicolumn{5}{|c|}{\makecell{\NL-complete \\ \small (Thm~\ref{thm:genreach-adam-only})}} \\
	\hline
\end{tabular}
	}
\caption{Complexity of solving different types of generalised reachability games.}
  \label{fig:complexity-table}
\end{figure}

\begin{figure}
  \centering
	\scalebox{.95}{
	\renewcommand{\arraystretch}{2.2}
	\setlength{\tabcolsep}{8pt}
	\begin{tabular}{|l|c|c|c| }
	\hline
	\small \diagbox{Players}{Targets} & \small All size $1$ & \small All size $\leq 2$ & \small General \\
    \hline
	2 Players & \makecell{\coNP-hard \\ \small (Thm~\ref{thm:maxgrcomplexity})} & \multicolumn{2}{|c|}{\makecell{\PSPACE-complete \\ \small (Thm~\ref{thm:maxgrcomplexity})}}\\
	\hline
	Only Eve & \makecell{\P \\ \small (Thm~\ref{thm:maxgrev})} & \multicolumn{2}{|c|}{\makecell{\NP-complete \\ \small (Thm~\ref{thm:maxgrev})}}\\
	\hline
	Only Adam & \makecell{\P \\ \small (Thm~\ref{thm:maxgrad})} & \multicolumn{2}{|c|}{\makecell{\coNP-complete \\ \small (Thm~\ref{thm:maxgrad})}} \\
	\hline
\end{tabular}
	}
\caption{Complexity of maximising the number of targets reached for different cases.}
\label{fig:maxgenreach-complexity-table}
\end{figure}

\begin{figure}
  \centering
	\scalebox{.95}{
	\renewcommand{\arraystretch}{2.2}
	\setlength{\tabcolsep}{8pt}
	\begin{tabular}{|l|c|c|c|c| }
	\hline
	\small \diagbox{Players}{Targets} & \small All size $1$ & \small All size $\leq 2$ & \small All size $\leq 3$ & \small General \\
    \hline
	2 Players & \makecell{\P \\ \small (Thm~\ref{thm:maxgrp2p})} & {open} & \multicolumn{2}{|c|}{\makecell{\PSPACE-complete \\ \small (Thm~\ref{thm:maxgrp2p})}}\\
	\hline
	Only Eve & \makecell{\P \\ \small (Thm~\ref{thm:maxgrev})} & \multicolumn{3}{|c|}{\makecell{\NP-complete \\ \small (Thm~\ref{thm:maxgrev})}}\\
	\hline
	Only Adam & \multicolumn{4}{|c|}{\makecell{\P \\ \small (Thm~\ref{thm:maxgrprad})}} \\
	\hline
\end{tabular}
	}
\caption{Complexity of maximising the number of reached targets that have to be declared before the game.}
  \label{fig:maxgenreachpromise-complexity-table}
  \end{figure}

\paragraph*{Related Work.}

%The starting point of this work is the study of \emph{generalised reachability games} by Fijalkow and Horn~\cite{fij-horn2010}.
%Their work revealed the surprising complexity of these games: while reachability games are \textsc{P}-complete \cite{Zer13,DBLP:journals/jcss/Immerman81}, even modest extensions to standard reachability objectives can lead to significant computational challenges. 
%This work has influenced the analysis of multi-objective games and strategy synthesis under complex constraints. 
%It also provides certain restrictions which make the problem easier, particularly by considering one-player variants, and by parametrising the problem by size of each target set.
%Fijalkow and Horn also exploit connections with the true quantified satisfiability (QSAT) problem, one of the standard \PSPACE-complete problems, to provide lower bounds, both in the general case and in some restricted cases. 
%An open problem stated in their work is the complexity of the problem when target sets have size $2$. 

Games with weighted multiple objectives have been studied by Kupferman and Shenwald~\cite{KupfermanS24}, showing how adding different goals and weights makes these games more complex and useful for modelling real systems.
In particular, their results can be used to provide strategies that maximise the number of objectives satisfied, even if not all of the objectives can be jointly satisfied. This yields a \PSPACE~upper bound for some of the optimisation variants of generalised reachability games we consider. However, they do not study the problem by considering the size of target sets as a parameter.

In the realm of satisfiability, El Halaby~\cite{Halaby2016} investigated the \emph{computational complexity of MaxSAT}, an optimisation variant of the satisfiability problem. 
Prior to this, Kohli et al.~\cite{Kohli-et-al-94} introduced and studied the \emph{Minimum Satisfiability Problem (MinSAT)}, establishing its NP-hardness and discussing its relevance in fields such as fault diagnosis and design verification. However, the optimisation variant of QSAT, a natural \PSPACE-complete problem, is relatively unexplored. To the best of our knowledge, while some algorithms for solving MAX-QSAT are known~\cite{IgnatievJM13}, the computational complexity for restricted classes of the problem have not been studied. As several lower bound proofs for generalised reachability games are obtained by reduction from QSAT, we expect that the analysis of optimisation variants of generalised reachability games provides more insight into the optimisation variants of QBF and vice-versa.
 
\section{Preliminaries}\label{sec:pre}
We use $\nat$ to denote the set of natural numbers. 
For $n \in \nat$, we use $[n]$ to denote the set $\{1,\dots,n\}$. 
We use $G = (V,E)$ to denote a directed graph with sets of vertices $V$ and edges $E \subseteq V \times V$. 
We often write $u \rightarrow v$ to denote $(u,v) \in E$. 
For $u \in V$, let $Succ(u) = \{v\in V \mid u \rightarrow v\}$; we generally assume
that $Succ(u)\neq\emptyset$ for all $u\in V$. 
We assume familiarity with graph theoretic notions such as strongly connected components (SCC) and the directed acyclic graph formed by SCC decomposition of a graph.%

We study two-player turn-based games played between players $\Ad$ and $\Ev$. 
Such games are played on directed graphs called game arenas. Formally, a \emph{game arena} $\Aa = (G, V_{\Ev}, V_{\Ad})$ is composed of a finite directed graph $G = (V, E)$ and a partition $(V_{\Ev}, V_{\Ad})$ of the vertex set $V$. A vertex in $V_{\Ev}$ (respectively $V_{\Ad}$) is controlled by $\Ev$ (respectively $\Ad$). For $v \in V$, let $\player{v}$ be the player controlling $v$ i.e. $\player{v} = p$ when $v \in V_p$ for $p \in \{\Ad, \Ev\}$.

A \emph{generalised reachability game} is a tuple $\game = (\Aa, s,\F)$, where\begin{itemize}
\item[--] $\Aa$ is a game arena,
\item[--] $s \in V$ is the start vertex,
\item[--] $\F = \{F_1, F_2, \dots, F_n\}$ is a set of $n$ target sets, where for each $i$, we have $F_i \subseteq V$.
\end{itemize}
%\red{There's an ugly space between where and first item above}

\Cref{fig:maxproblem} shows an example of a generalised reachability game where circle vertices belong to $\Ev$ and square vertices belong to $\Ad$. The start vertex is $s$. In this particular game, all target sets are singleton and are marked with doubled borders.  

The rules of a generalised reachability game are as follows: initially, a token is placed on the start vertex $s$. At each step, when the token is on vertex $u$, it is $\player{u}$'s turn to play: $\player{u}$ chooses a vertex $v$ from $Succ(u)$ and moves the token to $v$. The sequence of vertices starting with $s$ that is obtained this way, is called a \emph{play}. 
The objective of $\Ev$ is to move the token in such a way that the token visits at least one vertex from every target set $F_i \in \F$.
The goal of $\Ad$ is to prevent $\Ev$ from achieving her objective.
In order to achieve their respective goals, players can move tokens according to some \emph{strategy}.
When the token is on vertex $u$, $\player{u}$ moves the token based on the past history as described by the strategy. 
Formally,  a \emph{strategy} for player $p \in \{\Ad,\Ev\}$ is a function $ \sg_p : V^* \cdot V_{p} \to V $, such that for all $\pi=v_0 v_1 \ldots v_k\in V^* \cdot V_{p}$, we have $\sg_p(\pi)\in Succ(v_k)$.
Thus a strategy prescribes a valid move that should be taken when it is the respective player's turn.
A pair of strategies \( (\sigma_{\text{Eve}}, \sigma_{\text{Adam}}) \) induces a unique infinite play $\pi \in V^{\omega}$.

%\red{TODO : reformulate with lenght bound argument. reducing to finite horizon game.}
Given a generalised reachability objective $\F = \{F_1, F_2, \dots, F_n\}$,
a play $\pi=v_0 v_1\ldots$ is \emph{winning} for $\Ev$ if it visits each target set at least once, that is,
if for all $1\leq i\leq n$, there is $j \geq 0$ such that $v_j\in F_i$; otherwise, $\pi$ is won by $\Ad$. In principle, the game can continue for an infinite duration but for generalised reachability objectives, for every winning play $\pi=v_0 v_1\ldots$ there is some $k \geq 0$ such that $\Ev$ achieves her goal in the prefix $v_0 v_1\ldots v_k$.  

The \emph{winning region} of a player in a generalised reachability game is the set of vertices $v\in V$
such that the respective player has a strategy $\sigma$ such that the player wins every play that starts at $v$ and is 
compatible with $\sigma$.
Our primary interest in this paper is the complexity of deciding whether $\Ev$ has a strategy to achieve the generalised reachability objective, starting
from $s$ (that is, to decide whether $s$ belongs to the winning region of $\Ev$). The decision problem is as follows:

\begin{tcolorbox}[colback=gray!5!white,colframe=black!75!black,arc=0pt,outer arc=0pt]
$\boldsymbol{\genreach}$: Given a game $\game = (\Aa, s, \F)$, does $\Ev$ have a strategy to visit all $F_i$ in $\F$ starting from $s$?
\end{tcolorbox}

%\begin{thm}[\cite{fij-horn2010}]
%$\genreach$ is \PSPACE-complete. The \PSPACE-hardness holds even when $|F_i| = 3$ for each $F_i \in \F$. $\genreach$ is in $\P$ when $|F_i| = 1$ for each $F_i \in \F$.
%\end{thm}

%In this work, we particularly focus on  generalised reachability games parameterised by the size of target sets This includes cases where some target sets $F_i$ may have size $1$. In cases where our results rely on the number of singleton target sets, we simplify the notation for convenience. We write $\game = (\Aa,s,\F,T)$, where $|F_i| > 1$ for each $F_i \in \F$ and $T=\{t_1,\dots,t_m\} \subseteq V$ are the set of singleton targets In terms of the previous formulation, the set of targets is now $\F \cup \{\{t\}\mid t \in T\}$. 

\begin{figure}[!t]
 \centering
\scalebox{0.7}{
\begin{tikzpicture}[shorten >=0.7pt, node distance=1.6cm, on grid, auto, thick, >=stealth']
   \tikzstyle{state}=[circle, draw, fill=gray!10, minimum size=8mm]
   \tikzstyle{box}=[rectangle, draw, fill=gray!10, minimum size=8mm]

   \node[initial, state, initial text=] (s) [yshift=-0.8cm] {$s$};
   \node[box] (u) [below =of s] {$u$};
   \node[state,double, double distance=0.7mm] (u1) [below left=of u] {$u_1$};
   \node[state,double, double distance=0.7mm] (u2) [left=of u1] {$u_2$};
   \node[state,double, double distance=0.7mm] (u3) [below right=of u] {$u_3$};
   \node[state,double, double distance=0.7mm] (u4) [right=of u3] {$u_4$};
   \node[state,double, double distance=0.7mm] (v) [right=of s] {$v$};

   \path[->]

   (s) edge (u)
   (u) edge (u1)
   (u) edge (u3)
   (s) edge (v)
   (u1) edge [bend right=30] (u2)
   (u2) edge [bend right=30] (u1)
   (u3) edge [bend right=30] (u4)
   (u4) edge [bend right=30] (u3)
   (v) edge [loop right] (v);;
\end{tikzpicture}
}
\caption{A generalised reachability game with initial state $s$ and target sets $\F = \{\{u_1\},\{u_2\},\{u_3\},\{u_4\},\{v\}\}$.
$\Ev$ controls the circle nodes and $\Ad$ controls the square node. In this game, $\Ev$ cannot win the generalised reachability objective, but can ensure that at least two target sets are visited by choosing $u$ at $s$. When $\Ev$ has to declare the target sets she can guarantee to visit before the game starts, she can only promise one target set, namely $v$, and not any of the $u_i$'s.
}
\label{fig:maxproblem}
\end{figure}
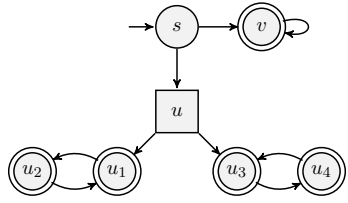
We also consider optimisation versions of the problem. In the game in \cref{fig:maxproblem}, all  targets are singleton with $\F = \{\{u_1\},\{u_2\},\{u_3\},\{u_4\},\{v\}\}$. Starting from $s$, $\Ev$ does not have a strategy to visit all targets.  However, if $\Ev$ chooses to move from $s$ to $u$, no matter where $\Ad$ moves at $u$, at least two target sets will be visited. This is the best $\Ev$ can do, since on choosing $v$, she can only visit one target set. With this in mind we consider the scenario where $\Ev$'s goal is to maximise the number of target sets she can visit defined. The decision problem is as follows: 

\begin{tcolorbox}[colback=gray!5!white,colframe=black!75!black,arc=0pt,outer arc=0pt]
$\boldsymbol{\maxgenreach}$: Given a game $\game= (\Aa, s, \F)$ and $ k \in \nat$, does $\Ev$ have a strategy to visit at least $k$ target sets from $\F$ starting from $s$?
\end{tcolorbox}

In \cref{fig:maxproblem}, $\Ev$ can force the play to visit at least two target sets, but she cannot choose two specific target sets to visit, since $\Ad$ makes the decision at $u$. On the other hand, $\Ev$ has a strategy to ensure that $v$ is always visited. Therefore the maximum number of target sets $\Ev$ can list out before the game that she can force to visit is one. We consider the variant where the objective of $\Ev$ is to select a maximum size collection of target sets before playing and then has to visit all of the selected target sets.
The decision problem is defined as follows: 
\begin{tcolorbox}[colback=gray!5!white,colframe=black!75!black,arc=0pt,outer arc=0pt]
$\boldsymbol{\maxgenreachpromise}$: Given a game $\game = (\Aa, s, \F)$, and $k\in \nat$, is there a set $\F' \subseteq \F$, with $|\F'| \geq k$  such that $\Ev$ has a strategy to visit all $F_i$ in $\F'$ starting from $s$?
\end{tcolorbox}

For the one-player variants of the game,  where only $\Ev$ plays%
, the  problems $\maxgenreach$
and $\maxgenreachpromise$ are equivalent. But in general, this may not hold, even for the one-player variant with $\Ad$ as the only player. To see this, consider a game obtained by modifying the game in \cref{fig:maxproblem}, where all vertices in the game belong to $\Ad$. The maximum number of target sets $\Ev$ can ensure for $\maxgenreach$ is $1$, whereas for the $\maxgenreachpromise$, the maximum is $0$, since there are no target sets in $\F$ that $\Ev$ can force to visit. 
$\Ad$ visits $v$ if $v$ is not promised by $\Ev$, otherwise goes to $u$ followed by a loop, say $u_1\to u_2 \to u_1$.

\paragraph{Attractors}
Reachability objectives require $\Ev$ to be able to enforce the play to enter target vertices, irrespective of what $\Ad$ plays. In this regard, the notion of \emph{attractor} is natural~\cite{Zer13,ZIELONKA1998135}.
For a set of (target) vertices $S\subseteq V$, the attractor of $S$ is the set of all vertices $u$ such that $\Ev$ has strategy to reach a vertex in $S$ when the game starts at $u$. Formally, the attractor of $S$, denoted by $\AtE(S)$ can be defined recursively as follows. 
\begin{align*}
\AtE^0(S) &= S \\
\AtE^{i+1}(S) &= \AtE^i(S) ~ \cup~ \{u \in V_{\Ev} \mid \exists v \in Succ(u),  v \in \AtE^i(S) \} \\ &~~~~ \cup~  \{u \in V_{\Ad} \mid \forall v \in Succ(u),v \in \AtE^i(S) \} \\
\AtE(S) &= \bigcup_i \AtE^i(S)
\end{align*}
Here, $\AtE^i(S)$ is the set of all vertices from which $\Ev$ can force plays to enter $S$ within $i$ steps.

We point out that attractor sets $\AtE(S)$ can be computed in time linear in the number of edges of the underlying graph $(V,E)$, that is, in time $\mathcal{O}(|E|)$, noting $|E|\leq |V|^2$.
For each vertex $v\in\AtE(S)$, player $\Ev$ can enforce that $S$ is visited when playing from $v$: for $v\in V$, let $i_v$ denote the least 
number such that $v\in \AtE^{i_v}(S)$. Then $\Ev$ intuitively can ensure that $S$ is reached from $v$ in at most $i_v$ steps. A witnessing positional strategy is obtained by
moving from $v\in \AtE^{i_v}(S)\cap V_{\Ev}$ to some $v'\in\AtE^{i_v-1}(S)\cap Succ(v)$, that is, by moving one step closer to $S$.

\section{Solving Generalised Reachability Games}

\noindent In this section, we study the complexity of the $\genreach$ problem. We give an almost complete picture of the complexity of $\genreach$ for different cases based on the size of target sets and the number of players. The results are summarised in \cref{fig:complexity-table}.

\subsection{Two-player Case}

We consider generalised reachability games with two players and show 
that their solution is \PSPACE-complete in general but becomes tractable if the size of target
sets is sufficiently restricted.

\subsubsection{Hardness of Generalised Reachability Games}

We first establish \PSPACE-hardness of solving two-player generalised reachability games in general. %This even when the underlying graph is a directed acyclic graph (DAG) with pathwidth 2. 

\begin{thm}\label{thm:pspace-hardness-general-case}
$\genreach$  is \PSPACE-complete. The \PSPACE-hardness holds even when $|F_i| = 3$ for each $F_i \in \F$ and the underlying graph is a directed acyclic graph with pathwidth $2$. 
\end{thm}
\begin{proof}
We reduce the \PSPACE-complete TQBF problem to generalised reachability games. It is known that TQBF is \PSPACE-hard even when each clause has 3 literals. 
The input of the reduction is a quantified Boolean formula over set of variables $x_1, \ldots, x_m$ of the form
\[\phi=Q_1 x_1. \ldots Q_i x_i.\ldots. Q_m x_m.\, c_1\land\ldots\land c_n\]
where each $Q_i \in \{\exists, \forall \}$, and each $c_j$ is a clause in disjunctive form with 3 literals (possibly negated), i.e., $c_j = l_{j_1} \lor l_{j_2} \lor l_{j_3}$ for some literals $l_{j_1}, l_{j_2}, l_{j_3}$. The reduced game is the generalised reachability game $G_\phi=(\mathcal{A}=(V,E),1,\{F_j~\mid ~ 1\leq j\leq n\} )$,
where $ V= \{i,x_i,\neg x_i\mid 1\leq i\leq m\} \cup\{\bot\} $. The edges of the arena are defined as follows, 
\begin{align*}
Succ(x_i)&=Succ(\neg x_i)=\{i+1\} ~\forall i ,~ 1\leq i < m &
Succ(i)&=\{x_i,\neg x_i\}  ~\forall i ,~ 1\leq i \leq m\\
Succ(x_m)&=Succ(\neg x_m)=\{\bot\} 
& Succ(\bot)&=\{\bot\} 
\end{align*}
The state $1$ is designated as the initial state and for all $j\leq n$, we have target sets $
F_j=\{l_{j_1}, l_{j_2}, l_{j_3}\}$. For instance, if the clause $c_j= x_1 \lor x_4 \lor \neg x_5$, then $F_j=\{x_1, x_4, \neg x_5\}$.
The partition of $V$ into $(V_\Ev,V_\Ad)$ is such that $i \in V_\Ad$ if $x_i$ is universally quantified, while $i \in V_\Ev$ if $x_i$ is existentially quantified. (The remaining vertices have a single successor, so that ownership does not matter.) See~\Cref{fig:qbf} for an example of the reduction.

There is a natural bijection between assignments of the variables $x_1, \ldots, x_m$ and plays in the game $G_\phi$: 
each path from the starting node $1$ to the sink assigns a truth value to each variable in the order determined by the graph structure. 
An assignment satisfies the formula $\phi$ if and only if the corresponding play visits all target sets $F_1, \ldots, F_n$ 
(i.e., satisfies the generalised reachability objective). 
Since the evaluation order of the variables is the same in the formula $\phi$ and in the game structure,
we conclude that $\Ev$ has a winning strategy in $G_\phi$ if and only if $\phi$ is true.
Moreover, the underlying graph of $G_\phi$ is a DAG with pathwidth $2$.

Now we show the \PSPACE upper bound.
Consider a $\game = (\Aa, s, \F)$ such that $\F=\{F_1,\dots, F_n\}$.
Let us first make a simple observation: if $\Ev$ has a winning strategy, then she has a winning strategy 
that visits each target set $F_i$ within $n \cdot |V|$ steps.
Indeed, if she can enforce to visit a subset of vertices, then she can enforce it within $|V|$ steps using attractor strategies.
Relying on this observation, we can simulate the game for up to $n \cdot |V|$ steps using an alternating Turing machine:
whenever a vertex belongs to $\Ev$, the corresponding state is disjunctive, and it is conjunctive if the vertex belongs to $\Ad$. 
A path of length $n \cdot |V|$ is accepted if it is winning, 
\textit{i.e.}, if it contains one vertex from each target set $F_i$. 
This machine accepts if and only if $\Ev$ wins $\game$, and works in polynomial time.
Since $\APTime = \PSPACE$, the result follows.

\end{proof}

We note that TQBF is in PTIME when all the clauses have 2 literals, and therefore our reduction does not work for deducing anything for the case where all target sets have size $2$.
%\textcolor{red}{add citation?}
%\begin{thm}\cite{fij-horn2010}
%Player $\Ev$ wins the game $G_\phi$ iff $\phi$ is true.
%\end{thm}

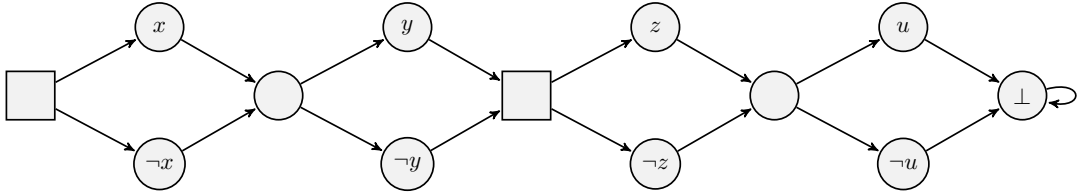
\begin{figure}[h!]
 \centering
\scalebox{0.8}{
\begin{tikzpicture}[shorten >=0.7pt, node distance=1.6cm, on grid, auto, thick, >=stealth']
   \tikzstyle{state}=[circle, draw, fill=gray!10, minimum size=8mm]
   \tikzstyle{box}=[rectangle, draw, fill=gray!10, minimum size=8mm]

   \node[box] (start) {};
   \node[state] (x) [above right=of start, xshift=1cm] {$x$};
   \node[state] (notx) [below right=of start, xshift=1cm] {$\neg x$};

   \node[state] (mid1) [right=of start, xshift=2.5cm] {};

   \node[state] (y) [above right=of mid1, xshift=1cm] {$y$};
   \node[state] (noty) [below right=of mid1, xshift=1cm] {$\neg y$};

   \node[box] (mid2) [right=of mid1, xshift=2.5cm] {};

   \node[state] (z) [above right=of mid2, xshift=1cm] {$z$};
   \node[state] (notz) [below right=of mid2, xshift=1cm] {$\neg z$};

   \node[state] (mid3) [right=of mid2, xshift=2.5cm] {};

   \node[state] (u) [above right=of mid3, xshift=1cm] {$u$};
   \node[state] (notu) [below right=of mid3, xshift=1cm] {$\neg u$};

   \node[state] (end) [right=of mid3, xshift=2.5cm] {$\bot$};

   \path[->]
   (start) edge (x)
   (start) edge (notx)
   (x) edge (mid1)
   (notx) edge (mid1)
   (mid1) edge (y)
   (mid1) edge (noty)
   (y) edge (mid2)
   (noty) edge (mid2)
   (mid2) edge (z)
   (mid2) edge (notz)
   (z) edge (mid3)
   (notz) edge (mid3)
   (mid3) edge (u)
   (mid3) edge (notu)
   (u) edge (end)
   (notu) edge (end)
   (end) edge [loop right] (end);
\end{tikzpicture}
}
 \caption{The generalised reachability game for the formula $\phi_1=
\forall x. \, \exists y. \, \forall z. \, \exists u. \, \big( (\neg x \lor \neg y\lor u) \land (x \lor \neg z) \land (\neg z \lor y) \big)$.}
\label{fig:qbf}
\end{figure}

\begin{exa}
\label{ex:qbf}
For an example of the reduction, consider the following QBF formula:
\[\phi_1=
\forall x. \, \exists y. \, \forall z. \, \exists u. \, \big( (\neg x \lor \neg y\lor u) \land (x \lor \neg z) \land (\neg z \lor y) \big).
\]
The reduced generalised reachability game $G_{\phi_1}$ is shown in ~\Cref{fig:qbf}.
$\Ev$ takes care of existential quantification in $\phi_1$, and $\Ad$ of universal quantification. The clauses correspond to the sets \( F_i \); in this example, we have 
\[
F_1 = \{\neg x, \neg y, u\}, \quad F_2 = \{x, \neg z\}, \quad F_3 = \{\neg z, y\}.
\]

\end{exa}

\subsubsection{Generalised Reachability with Small Target Sets}

Next we show that the solution of generalised reachability games becomes tractable
when the sizes of the individual target sets are sufficiently restricted. 

In order to win a generalised reachability game, $\Ev$ needs to be able to
force visits to all target sets, that is, $\Ev$ loses from all nodes $s\in V$ such that
$s\notin\textstyle\bigcap_{1\leq i \leq k} \AtE(F_i)$. Containment in all individual attractors, however,
is not sufficient for winning: $\Ev$ has to ensure that all target sets are visited (in some order),
so once a target set is visited, she has to be able to visit all remaining targets.

Consider a game with two singleton target sets $F_0=\{t_0\}$ and $F_1=\{t_1\}$. Any winning strategy forces
a visit to both $t_0$ and $t_1$ (in some order).
Depending on the behaviour of $\Ad$, $\Ev$ may be able to first visit $t_0$ and
then $t_1$, or vice versa.
Let $i\in\{0,1\}$ and $s\in V$.
By the definition of attractors, whenever $\Ev$ can attract from $s$ to the single node $t_i$ and from $t_i$ to $t_{1-i}$,
then she can attract from $s$ to $t_{1-i}$ by simply concatenating the two attractor strategies. 
Note that the play may visit $t_{1-i}$ already while going from $s$ to $t_i$, but $\Ev$ can guarantee a visit to $t_{1-i}$ 
after visiting $t_i$ as well. It follows that if $\Ev$ wins, then she has a strategy to
visit the targets in a single fixed order, no matter how $\Ad$ plays. Intuitively,
any such strategy first visits the target from where $\Ev$ can enforce a visit to the other target; 
if this is the case for both targets, then the order in which they are visited can be picked arbitrarily.
Thus $\Ev$ wins $s$ if and only if there is $i\in\{0,1\}$ such that she can attract
from $s$ to $t_i$, and then from $t_i$ to $t_{1-i}$.
This is the case if and only if 
there is $i\in\{0,1\}$ such that $\AtE(F_i)\subseteq\AtE(F_{i-1})$ and
$s\in \AtE(F_i)$.

This argument extends to generalised reachability games with $k$ targets that all are singleton sets. Such
games can be solved by computing the attractors $\AtE(F_{i})$ to all individual targets and then checking whether these attractors are ordered by set inclusion.
If this is the case, then the winning region of $\Ev$ is $\textstyle\bigcap_{1\leq i \leq k}\AtE(F_{i})$, otherwise,
the winning region of $\Ev$ is the empty set.
\begin{thm}\label{thm:ptime-singleton-2player}
$\genreach$ is in \P~when all target sets are singleton sets. \end{thm}
\begin{proof}
Let $\game = (\Aa,s,\F=\{F_1,\ldots,F_n\})$ be a generalised reachability game
such that $|F_i|=1$ for all $1\leq i\leq n$. 
For each $i\in\{1,\ldots,n\}$, compute $A_i=\AtE(F_{i})$. This can be done in time
$\mathcal{O}(n|E|)$. Check whether there are indices $i,j\in\{1,\ldots,n\}$
such that $A_i\not\subseteq A_j$ and $A_j\not\subseteq A_i$.
This can be checked in time $\mathcal{O}(n^2 \log |V|)$. If such indices exist,
then $\Ad$ wins by any strategy that plays arbitrarily until either $F_i$ or $F_j$ is visited,
and then forever avoids the other target set; such strategies proceed in two (positional) stages 
so that $2$ memory values suffice. Otherwise, $\Ev$ wins by forcing visits
to all sets $F_i$ according to the inclusion preorder on the sets $A_i$. There is a strategy for this 
that uses at most $n$ memory values, at each point identifying the current target.
\end{proof}

This result generalises to the case where there is a fixed number of non-singleton target sets. 
For illustrative purposes, we first consider the version with $n+1$ target sets of which $n$ are singleton targets $F_1=\{t_1\}, \dots, F_n=\{t_n\}$, and one target set $F_{n+1}$ has more than one element (that is, $|F_{n+1}| > 1$).

\begin{thm}\label{thm:genReachSingletons}
$\genreach$ is in \P~when one target set has size at least $2$ and all other target sets are singleton sets.  
\end{thm}

\begin{proof}

Let $\game = (\Aa,s,\F=\{F_1,\ldots,F_n, F_{n+1}\})$ be a generalised reachability game
such that $|F_{n+1}|>1$ and $F_i=\{t_i\}$ for all $1\leq i\leq n$.
Let $A_1,\ldots,A_n$ denote the attractor sets to the singleton target sets, that is $A_i=\AtE(\{t_i\})$ for $1\leq i \leq n$.
We point out that $t_i\in A_i$.
We claim that $\Ev$ wins the game $\mathcal{G}$ from $s$ if and only if
\begin{enumerate}
\item the sets $A_i$ form a total preorder under set inclusion;
\item $s\in \textstyle\bigcap_{1\leq i\leq n} A_i$; and 
\item there is some $0\leq j \leq n$ such that $t_{j+1}\in \AtE(A_{j} \cap F_{n+1})$, where $t_{n+1}= s$.
\end{enumerate}
The attractor computations and the check whether they form a total preorder can be implemented in polynomial time. The same holds
for the check whether $s\in \textstyle\bigcap_{1\leq i\leq n} A_i$ and whether one of the target states $t_{j+1}$ is contained in the attractor to $A_j\cap F_{n+1}$. 
Hence the Theorem follows from the claim.

\begin{figure}[h!]
 \centering

\begin{tikzpicture}[>={[inset=0,angle'=27]Stealth}, SharpSquiggly/.style={
            decorate,
            decoration={zigzag, segment length=4, amplitude=0.9},
        }]
\draw[thick,fill=gray!10,rounded corners=10,shift={(0.15,0.15)}]
    (-3,-2) rectangle (4.5,1.7);
\draw [thick,fill=gray!50](0,0) ellipse (2.5cm and 1.4cm);
\draw [thick,fill=gray!35](0,0) ellipse (2cm and 1cm);
\draw [thick,fill=gray!20](0,0) ellipse (1.5cm and .6cm);
\draw [thick,fill=gray!30](3,-0.5) ellipse (1.3cm and 0.8cm);
\draw [thick](0,0) ellipse (2cm and 1cm);
\draw [thick](0,0) ellipse (2.5cm and 1.4cm);
    
\node at (-1.1, -.1){\footnotesize{$A_3$}};
\node at (-1.35, -0.45){\footnotesize{$A_{2}$}};
\node at (-1.6, -0.8){\footnotesize{$A_{1}$}};
\node at (-2, -1.5){\footnotesize{$A_{0}$}};
\node at (3, -.5){\footnotesize{$F_{0}$}};
\node at (-.5,-.2) [circle,fill,inner sep=1pt]{};
\node at (-.3,-.3){\footnotesize{$s$}};
\node at (.5,.3) [circle,fill,inner sep=1pt]{};
\draw[->]  (-.5,-.2) decorate[SharpSquiggly]{ -- (.5,.3) } ;
\node at (.7,.25){\footnotesize{$t_3$}};
\draw[->]  (.5,.3) decorate[SharpSquiggly]{ -- (.8,-.6) } ;
\node at (.8,-.6) [circle,fill,inner sep=1pt]{};
\node at (.6,-.65){\footnotesize{$t_2$}};
\draw[->]  (.8,-.6) decorate[SharpSquiggly]{ -- (2.2,-.3) } ;
\node at (2,.6) [circle,fill,inner sep=1pt]{};
\node at (1.8,.7){\footnotesize{$t_1$}};
\draw[->]  (2.2,-.3) decorate[SharpSquiggly]{ -- (2,.6) } ;
\end{tikzpicture}
\caption{Example construction of winning strategy from attractor sets, $k=4$}
\label{fig:genreach}
\end{figure}
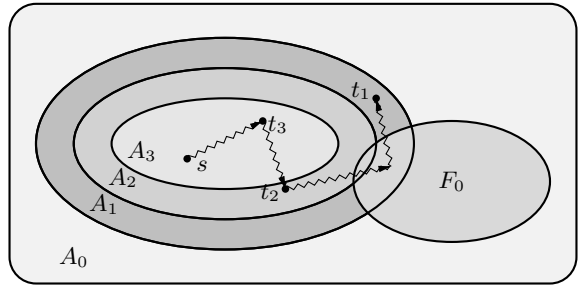

For one direction of the claim, assume that the three items hold.
As the order in which the individual targets are visited is irrelevant, we can reorder them. Without loss of generality, we assume 
\ $A_n\subseteq \ldots \subseteq A_1$. Let the initial state in $\mathcal{G}$ be denoted by $t_{n+1}=s$ so that $t_{n+1}\in A_n$
by item ($2$).
Additionally, let $A_0$ denote the set of all states.
We construct a winning strategy for $\Ev$ as follows:

1. For each $i \in \{1,\ldots, n\}$ with $i\neq j+1$: move from $t_{i+1}$ towards $t_i$. Since $t_{i+1}$ is in the attractor $A_i$, player $\Ev$ can enforce that $t_i$ is eventually reached.

2. For the state $t_{j+1}$: move towards the set $A_{j} \cap F_{n+1}$; this is possible since $t_{j+1}\in \AtE(A_{j} \cap F_{n+1})$. If $j = 0$,  then the strategy terminates successfully once $A_{0} \cap F_{n+1}$ is reached. Otherwise, proceed by moving from $A_{j} \cap F_{n+1}$ on to $t_{j}$ (which is possible since 
$A_{j} \cap F_{n+1}\subseteq A_j$), and finish the remaining sequence of targets from $t_j$.

This strategy uses at most $n+1$ memory values to identify the current target set and ensures that all target sets, including $F_{n+1}$, are eventually visited.
Figure~\ref{fig:genreach} shows an example of the strategy construction, where $s=t_4$ and we assume that $j=1$, that is, that $t_2$ is contained in $\AtE(A_1\cap F_{n+1})$.

For the converse direction, assume that one of the three items does not hold.

If the attractors do not form a total preorder, then there exist indices 
$i,j$ such that neither $A_i \subseteq A_j$ nor $A_j \subseteq A_i$ holds. 
Then player $\Ad$ wins by playing arbitrarily until $t_i$ or $t_j$ is reached, and then avoiding the other target forever.
Two memory values suffice to implement strategies with this behaviour.

We consider the remaining case in which the attractors form a total preorder, but $s\notin\textstyle\bigcap_{1\leq i \leq n} A_i$ or we have $t_{j+1}\notin \AtE(A_{j} \cap F_{n+1})$ for all $j$.
In the former case, player $\Ev$ cannot attract from $s$ to some target $t_i$. Thus player $\Ad$ wins from $s$ using a positional strategy that avoids $t_i$ forever.
In the latter case, player $\Ad$ wins by staying within the attractors $A_i$ but avoiding the sets $A_{j} \cap F_{n+1}$ for all $j$. While this strategy may visit all singleton targets $t_i$, it avoids the set $F_{n+1}$ forever so that $\Ad$ wins.
\end{proof}

%\subsection{Generalised Reachability with Mostly Singleton Targets}

Next we consider the case of generalised reachability games with a total of $n+k$ target sets: $n$ singleton targets and $k$ target sets of size at least $2$.

\begin{thm}\label{thm:fpt-2-player}
$\genreach$ can be solved in time $\mathcal{O}(mn2^k)$ for generalised reachability games with $m$ edges, $n$ singleton target sets and $k$ target sets of size at least $2$.  
\end{thm}
\begin{proof}
Let $\mathcal{G} = (\mathcal{A}, s, \{F_1,\ldots,F_{k+n}\})$ be a generalised reachability game with parameters as stated in the claim,
that is $|F_i|>1$ for $1\leq i\leq k$ and $F_{k+j}=\{t_j\}$ for $1\leq j\leq n$.
First, we observe that a necessary condition for player $\Ev$ to win $\mathcal{G}$ is that 
$\Ev$ wins $(\mathcal{A}, s, \{F_{k+1},\ldots,F_{k+n}\})$.
\iffalse
Here, \texttt{GEN-REACH}$(G, s, T, F)$ denotes a \emph{generalised reachability game}, where:
\begin{itemize}
  \item $G = (V, (V_\circ, V_\square), E)$ is the game arena, consisting of a finite directed graph with vertex set $V$, edge set $E$, and a partition of $V$ into positions controlled by $\Ev$ ($V_\circ$) and $\Ad$($V_\square$);
  \item $s \in V$ is the start vertex;
  \item $F = \{F_1, \dots, F_k\}$ is a collection of $k$ target sets, where each $F_i \subseteq V$;
  \item The generalised reachability objective is to reach at least one vertex from each $F_i$. Formally, the set of winning plays for Eve is:
  \[
  \texttt{GenReach}(F_1, \dots, F_k) = \left\{ \pi = v_0 v_1 v_2 \dots \,\middle|\, \forall i \in \{1, \dots, k\},\, \exists p_i \in \mathbb{N},\, v_{p_i} \in F_i \right\}.
  \]
\end{itemize}
\fi
If player $\Ev$ wins $ (\mathcal{A}, s, \{F_1,\ldots,F_{k+n}\})$,
we hence can follow the proof of Theorem~\ref{thm:genReachSingletons} by
assuming without loss of generality a total order on the singleton target states as $t_1,\dots,t_n$ where $s = t_0 \in \AtE{(t_1)}$ and $t_i \in \AtE{(t_{i+1})}$ for $1\leq i<n$. Then let $T$ denote $\{t_0,t_1,\ldots,t_n\}$.

Next we transform the game arena of $\mathcal{G}$ in order to treat the non-singleton target sets $\mathcal{F}$. To this end, 
let $V$ denote the set of vertices of $\mathcal{A}$ and 
consider the game arena $\hat{\mathcal{A}}$ over $\hat{V} = V \times 2^{[k]}$, that is, $V$ is augmented with the memory structure $2^{[k]}$, storing the set of target sets from
$\{F_1,\ldots,F_k\}$ that have been visited so far. States $(u,S)\in\hat{V}$ are owned by the player owning $u$ in $\mathcal{A}$ and we have an edge $(u,S) \rightarrow (v,S')$ in $\hat{\mathcal{A}}$ iff (i) $u \rightarrow v$ is an edge in $\mathcal{A}$ and (ii) $S'= S \cup \{j\in\{1,\ldots,k\} \mid v \in F_j\}$.

Next, we inductively define a sequence of sets $D_i$ by putting $D_{n+1} = V \times [k]$ and, for $0\leq i\leq n$,
$D_i = \{(u,S) \in A_{i+1} \mid u = t_{i}\}$, where $A_{i+1} = \AtE(D_{i+1})$.

We claim that $\Ev$ wins ${\mathcal{G}}$ if and only if $(s,\emptyset)\in D_0$.

For one direction, let $(s,\emptyset)\in D_0$. Then $\Ev$ can just follow her individual attractor strategies one after another, which results in visits
to all target states $t_i\in T$, as well to all target sets $F_i\in \F$; the latter is the case since the values of the auxiliary memory indicate the 
target sets that have been visited so far, and since the constructed strategy ensures that a game node from
$D_{n+1}$, that is, with auxiliary memory value $[k]=\{1,\ldots,k\}$ is eventually reached.

For the converse direction, let $(s,\emptyset)\notin D_0$.
We show that $\Ad$ wins $\mathcal{G}$
%We first assume that the attractors of all singleton target sets in $\mathcal{G}$ are different.
by constructing a winning strategy for $\Ad$ in $\mathcal{G}$ as follows.
We note that $\Ad$ can win by preventing $\Ev$ from visiting the individual target vertices $t_i$ in the order
prescribed by the partial order on their attractors in $\mathcal{G}$.
For this, he wins if at some point, the first component of $\hat{\mathcal{A}}$ (the original game node in $\mathcal{A}$) falls outside of the attractor for the next target vertex.
Plays for which this does not happen may eventually visit all singleton target sets. In this case, it still follows from $(s,\emptyset)\notin D_0$ that
$\Ad$ can avoid, after visiting all singleton targets, at least one of the larger target sets $F_i$ for $1\leq i\leq k$.

With this in mind, $\Ad$ can just follow a strategy in $\mathcal{G}$ that attempts to prevent that the induced play on $\hat{\mathcal{A}}$ ever reaches $D_1$ from vertices of the form $(t_0,S) \notin D_0$; if this fails, $\Ad$ uses his strategy to prevent states of the form $(t_1,S) \notin D_1$ to reach $D_2$, and so on.
Since $(s,\emptyset)\notin D_0$, $\Ad$ can use this strategy to ensure
that whenever a vertex $u$ is reached by a play in $\mathcal{G}$ that visits all target sets $F_i\in \mathcal{F}$ (which corresponds to reaching the vertex $(u,[k])\in D_{n+1}$ in
the corresponding play on $\mathcal{A}'$), then at least one of the target states $t_i$ has not been visited by the play so far. Thus the described strategy indeed is winning for $\Ad$, as required.

%We note that if two or more target states have the same attractor, then they are neighbours $t_i,\ldots,t_j$ in the partial order on attractors.
%For deciding reachability, we can then just keep one such state, keeping in mind that we can reach all other such states and return to any of these target states afterwards.

Regarding runtime complexity, the proposed solution algorithm computes $n+2$ sets $D_i$ and $A_i$; both are subsets of $V\times 2^{[k]}$. Given $A_{i+1}$, the computation of a single set $D_i$ can be done in time linear in $m \cdot 2^k$.
The computation of a single attractor set $A_i$ over a graph with $m \cdot 2^k$ edges takes time $\mathcal{O}(m \cdot 2^k)$. Hence the overall algorithm can be implemented to run in time
$\mathcal{O}(nm2^k)$.
\end{proof}

This in particular provides fixed parameter tractability.

\begin{cor}
$\genreach$ is in \P~when the number of target sets that are not singleton is logarithmic in the size of the game.
\end{cor}

\subsection{One Player Case}
\paragraph{One Player Case with $\Ev$}
We consider the case of generalised reachability games in which
$\Ev$ controls all vertices. The solution of such games is $\NP$-complete in general,
but becomes tractable if sizes of target sets are suitably restricted.

\begin{thm}\label{thm:genreach-eve-only}
  Let $\game =(\Aa,s,\F)$ be a generalised reachability game with all vertices controlled by $\Ev$.
Then $\genreach$ is 
\begin{enumerate}
  \item \NP-complete in general and is \NP-hard even when $|F_i| = 3$ for each $F_i \in \F$;
  \item in $\P$ if for all $F_i\in \F$, $|F_i| \leq 2$;
  \item $\NL$-complete if for all $F_i\in \F$, $|F_i| = 1$.
\end{enumerate}
\end{thm}

\begin{proof}
\begin{enumerate}

\item 
%First, we prove the \NP upper bounds.
Player $\Ev$ has control over all branching and hence can guess a path that satisfies the winning condition. If $\Ev$ wins, then she has a winning strategy that wins within $|V| \cdot |\F|$ steps. The algorithm guesses a path of length $|V| \cdot |\F|$ and checks whether it is winning. This gives an $\NP$ upper bound.

For the lower bound, in our previous reduction for \PSPACE-hardness in the proof of Theorem \ref{thm:pspace-hardness-general-case}, consider the case where all variables in the original formula are quantified existentially. 
Then the problem corresponds to 3-SAT, which is $\NP$-complete.
Resulting games are one-player games, \textit{i.e} all vertices belong to Eve,
hence solving one-player generalised reachability games is $\NP$-hard even when $|F_i| = 3$ for each $F_i \in \F$.

\item 
If all target sets are of size at most 2, then we can encode the winning condition as a 2SAT formula and check its 
satisfiability in polynomial time:

Similar to the previous subsection, we consider the preorder on the set of vertices $V$ in $\mathcal{G}$
defined by $v \preceq v'$ if $v \in \AtE(v')$.
Note that in the case of one-player arenas, $v \in \AtE(v')$ reduces to ``there is a path from $v$ to $v'$''.

Let $F_i = \set{x_i,y_i}$ be the reachability sets, and $v_0$ be the starting vertex. 
We assume without loss of generality that there is a path from $v_0$ to every $F_i$ (that is, either to $x_i$ or $y_i$),
otherwise $\Ev$ cannot win. 
(This property is easily checked in deterministic polynomial time.)
A first statement is as follows: $\Ev$ wins from $v_0$ if and only if 
there exist $v_1,\ldots,v_k$ target vertices such that
\begin{enumerate}
	\item[1.] for all $0 \leq i \leq k-1$, $v_i \preceq v_{i+1}$ and
	\item[2.] each target set $F_i$ is represented in $\set{v_1,\ldots,v_k}$.
\end{enumerate}
We turn this condition into a boolean formula where clauses have size $2$.
We consider the $2\cdot k$ variables $X_i$ and $Y_i$,
that correspond to vertices $x_i$ and $y_i$.
We define the formula $\phi$:
$$\underbrace{\bigwedge_{\ } \set{(\neg X \vee \neg Y) \mid 
\textrm{ if } x \not\preceq y \textrm{ and } y \not\preceq x}}_{(a)}
\wedge \underbrace{\bigwedge_{i} (X_i \vee Y_i)}_{(b)},$$
where $x,y$ ranges over target vertices (that is, vertices from $F_i$ for some $i$).

We argue that $\Ev$ wins from $v_0$ if and only if $\phi$ is satisfiable.
Assume $\Ev$ wins from $v_0$: let $v_1,\ldots,v_k$ be as in the previous statement,
and set the corresponding variables to true and the others to false, we claim that the formula $\phi$ is satisfied.
Indeed, condition 2. ensures that the clauses under-braced $(b)$ are satisfied,
and for the clauses under-braced $(a)$, let $x,y$ such that $x \not\preceq y$ and $y \not\preceq x$.
If $x$ is one of the $v_i$'s, then $y$ cannot be, so $\neg X \vee \neg Y$ holds.
Conversely, assume that $\phi$ is satisfiable.
The clauses under-braced $(a)$ ensure that the order $\preceq$ is total over vertices set to true.
The clauses under-braced $(b)$ ensure that at least one vertex from each reachability set is set to true.
Combining those two statements, we reach the condition stated above.

The latter allows to decide in polynomial time whether $\Ev$ wins from $v_0$ 
by checking the formula $\phi$ for satisfiability.
\item 
If  all target sets are singletons, then we can check in logarithmic space if there exist 
targets $t_i,t_j\in T$ such that $t_i$ is not reachable from $t_j$ and $t_j$ is not reachable from $t_i$. 
If no such pair of targets exists, then $\Ev$ wins by going through all targets in order of reachability. If such a pair of targets exists, 
then any play must first reach one of the two targets, say $t_i$, and then it cannot reach $t_j$ and hence $\Ev$ loses.

\end{enumerate}
\end{proof}

\paragraph{One Player Case with $\Ad$}
Next, we consider the special case of generalised reachability games in which $\Ad$ controls all vertices.
We show that solving such games is $\NL$-complete.
%For one-player games with $\Ev$, $\genreach$ is known to be $\NP$-complete in general, and in $\P$ when $|F_i| \leq 2$ for all $F_i \in \F$~\cite{fij-horn2010}.
%If $|F_i|=1$ for all $F_i\in \F$, \cite[Theorem 5]{austin2025} can be modified to obtain $\NL$-completeness. 
%The general case was shown to be in $\P$ in~\cite{fij-horn2010}. We provide an improved complexity bound by showing that the problem is in fact $\NL$-complete. 
\begin{thm}\label{thm:genreach-adam-only}
Let $\game =(\Aa,s,\F)$ be a generalised reachability game with all vertices controlled by $\Ad$. Then,
$\genreach$ is $\NL$-complete. The $\NL$-hardness holds even when $|\F|=1$. %
\end{thm}
\begin{proof}
First, we show containment in $\NL$. $\Ad$ wins if he has a strategy to ensure that some $F_i \in \F$ is never visited. In other words, $\Ad$ wins from vertex $v$ iff there is some $i$ such that $v \not \in \AtE(F_i)$. %
Observe that, if $\Ad$ wins, there is always a winning strategy of $\Ad$ that produces a \emph{lasso} play, i.e. a play of the form $v_0, \dots ,v_i, \dots ,v_m$ where $m \leq |V|$, $v_i = v_m$ and no vertex is repeated except $v_i = v_m$. This is because, if a winning strategy produces a play with more than one loop, then $\Ad$ has another winning strategy where at least one of those loops is not taken. We call such a play an $(v_0,v_i,v_i)$ lasso.  
Based on this, we provide an $\NL$ algorithm for the complement decision problem, i.e. for checking whether $\Ad$ has a winning strategy.
Since $\NL = \coNL$, this gives an $\NL$ upper bound. 

We non-deterministically guess a $(s,t,t)$ lasso and a set $F_i$ such that no vertex from the lasso is in $F_i$.
We guess the lasso by guessing successors step by step starting from $s$. At each step we check that the vertex is not in $F_i$. We also non-deterministically guess the vertex $t$ at some step and store it. We also maintain a counter to store the length of the play so far. The algorithm terminates when $t$ is repeated and the length is no more than $n$. This algorithm uses logarithmic space.

\begin{figure}[h!]
 \centering
\scalebox{0.7}{
\begin{tikzpicture}[shorten >=0.7pt, node distance=1.6cm, on grid, auto, thick, >=stealth']
   \tikzstyle{state}=[circle, draw, fill=gray!10, minimum size=8mm]
   \tikzstyle{box}=[rectangle, draw, fill=gray!10, minimum size=8mm]

   \node[state] (s) [yshift=-0.8cm] {$s$};
   \node[state] (u) [below right=of s] {$u$};
   \node[state] (t) [above right=of u] {$t$};

   \node[box] (s1) [right=of start, xshift=4cm,yshift=0.8cm] {$s,1$};
   \node[box] (u1) [below=of s1] {$u,1$};
   \node[box] (t1) [below=of u1] {$t,1$};

   \node[box] (s2) [right=of s1, xshift=1cm] {$s,2$};
   \node[box] (u2) [below=of s2] {$u,2$};
   \node[box] (t2) [below=of u2] {$t,2$};

   \node[box] (s3) [right=of s2, xshift=1cm] {$s,3$};
   \node[box] (u3) [below=of s3] {$u,3$};
   \node[box] (t3) [below=of u3] {$t,3$};

   \node[box] (s4) [right=of s3, xshift=1cm] {$s,4$};
   \node[box] (u4) [below=of s4] {$u,4$};
   \node[box] (t4) [below=of u4] {$t,4$};

   \node[box,double, double distance=0.7mm] (bot) [above=of s2, xshift=1cm] {$\bot$};
   \node[box] (top) [below=of t2, xshift=1cm] {$\top$};

   \path[->]

   (s) edge [bend left=20] (u)
   (u) edge [bend left=20] (s)
   (u) edge (t)
   (u) edge [loop right] (u)
   (t) edge [loop right] (t)

   (s1) edge [dashed] (u2)
   (u1) edge (u2)
   (u1) edge (s2)
   (u1) edge (t2)
   (t1) edge (t2)
   
   (s2) edge (u3)
   (u2) edge  (u3)
   (u2) edge (s3)
   (u2) edge [dashed] (t3)
   (t2) edge (t3)

   (s3) edge (u4)
   (u3) edge (u4)
   (u3) edge (s4)
   (u3) edge (t4)
   (t3) edge (t4)

   (t1) edge (top)
   (t2) edge (top)
   (t3) edge [dashed] (top)
   (t4) edge (top)

   (s1) edge (bot)
   (s2) edge (bot)
   (s3) edge (bot)
   (s4) edge (bot)
   (u1) edge [bend left=10] (bot)
   (u2) edge [bend right=10] (bot)
   (u3) edge [bend left=10] (bot)
   (u4) edge [bend right=10] (bot)

   (top) edge [dashed, loop below] (top)
   (bot) edge [loop above] (bot);
\end{tikzpicture}
}
 \caption{Example of reduction from $s-t$ reachability to one-player generalised reachability.}
    \label{fig:nl}
\end{figure}
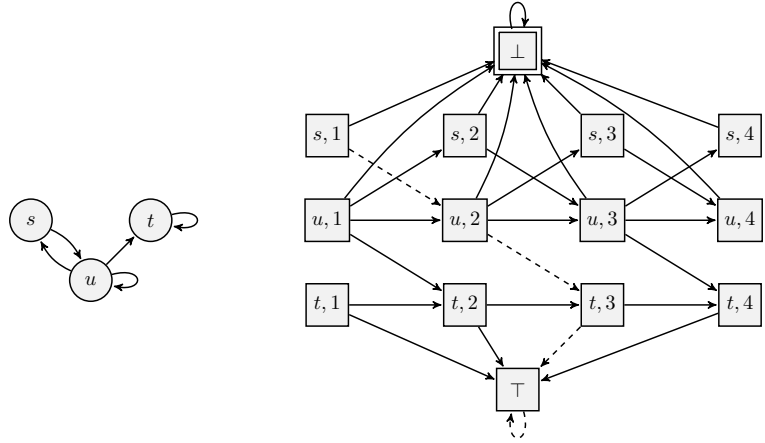

For the lower bound, we provide a reduction from the $\NL$-complete $s-t$ reachability problem. The $s-t$ reachability problem asks whether a vertex $t$ is reachable from a vertex $s$ in a given graph $H = (V_H,E_H)$ with $s,t \in V_H$. We construct a game $\Gg$ such that $t$ is reachable from $s$ in $H$ if and only if $\Ad$ has a winning strategy.

Let $|V_H| = n$. The game $\Gg$ uses the arena $\Aa$ with underlying graph $G = (V,E)$ over $V = V_H \times [n+1] \cup \{ \top, \bot\}$, all of which are owned by $\Ad$. If $(u,v) \in E_H$, then for every $1\leq i \leq n$, in $E$ we have $(u,i) \rightarrow (v,i+1)$. For every $i$ we have $(t,i) \rightarrow \top$. For every $v \in V\setminus \{t\}$ and for each $i$, we have $(v,i) \rightarrow \bot$. The vertices $\top$ and $\bot$ have self loops. The start vertex is $(s,1)$ and $\Gg$ has only one singleton target set $T = \{\bot\}$. %
We note that $\Gg$ can be constructed in logspace from $H$.
Figure~\ref{fig:nl} shows an example of this reduction for a graph with three vertices; the dashed arrows indicate a strategy for $\Ad$ to spoil reachability of $\{\bot\}$ in the reduced game,
corresponding to the fact that $t$ is reachable from $s$ in the graph.

We claim that in general, $t$ is reachable from $s$ in $H$ iff $\Ad$ wins the game $\Gg = (\Aa,(s,1),\{\bot\})$. %
To prove the claim, we observe that all infinite paths in $G$ end up in exactly one of the self loops at $\top$ or $\bot$. If $t$ is reachable from $s$, then $\Ad$ has a path from $(s,1)$ to $(t,i)$ for some $i$ and then can move to $\top$ and hence wins the game, i.e. $\Ad$ has a winning $((s,1),\top,\top)$ lasso. If $t$ is not reachable from $s$, then all paths starting from $(s,1)$ end up at $\bot$, in which case $\Ad$ loses. 
\end{proof}

%\vspace{10pt}
%\textcolor{red}{Move the following to the opimisation section?} 
%Makes Sense - moved.
%\begin{rem}
%For one-player games with $\Ad$, even for the optimisation variants of the problem, i.e. the $\maxgenreach$ and $\maxgenreachpromise$ problems, $\Ad$ always has an optimal strategy that produces a lasso. This follows from the fact that it is always better for $\Ad$ to encounter fewer distinct vertices, and consequently fewer target sets. 
%\end{rem}

\section{Memory Requirements for Winning Generalised Reachability Games}\label{sec:mem}
%!TEX root=./main.tex

In this section, we study the following question: 
\emph{How much memory is required to win generalised reachability games?}
\newline
\newline
First we define memory machines associated to strategies. 
For a  game $\game = (\Aa,s,\F = \{F_1,\dots,F_k\})$, a memory machine is a tuple $\M = (M, m_0, \upd)$, where 
$M = 2^{[k]}$ is the set of memory states recording which target sets have been already visited, $m_0 = \set{i \mid s \in F_i}$ is the initial memory state, and $\upd : M \times E \to M$ is the memory update function defined by $\upd(S,(v,v')) = S \cup \set{i \mid v' \in F_i}$. Note that, this is very similar to the auxiliary memory structure used in the proof of Theorem~\ref{thm:fpt-2-player}.

Given $\game$ and $\M$, one can construct the natural product game $\game \times \M$ over the set of nodes $V \times M$ and edges following the update function $\upd$. Every play in $\game$ naturally corresponds to a unique play in $\game \times \M$ and vice versa. There is an equivalence between winning the generalised reachability game $\game$ and winning in $\game \times \M$ with a reachability objective of reaching $V \times \{[k]\}$, which gives the following key observation.
\begin{lem}\label{lem:product-game-equivalence}
A play in $\game$ from $s$ is winning for $\Ev$ if and only if
the corresponding play is winning for $\Ev$ in the game $\game \times \M $
starting from $(s,m_0)$ with reachability objective of reaching $V \times \{[k]\}$.
\end{lem}
\begin{proof}
By construction of $\upd$, after any play prefix $v_0 v_1 \ldots v_j$ in $\game$, the memory component of the corresponding product state $(v_j, S)$ in $\game \times \M$ records exactly the set of targets visited so far. Therefore, the play in $\game$ eventually visits all target sets (winning condition for $\Ev$) if and only if the corresponding play in $\game \times \M$ eventually reaches a state with memory component $[k]$, i.e.\ reaches $V \times \{[k]\}$.
\end{proof}
Hereafter we will use $\game \times \M$ to denote the game with the reachability objective (of $\Ev$) of reaching $V \times \{[k]\}$. Note that a winning strategy for any player in this game is always positional. We will now make the following crucial observation about the structure of winning strategies for $\Ad$ in $\game \times \M$.

\begin{lem}
\label{lem:adam-mem-downward-closed}
For $v \in V$ and $S \subseteq [k]$, if $\Ad$ wins the game $\game \times \M$ starting from $(v,S)$, then for every $S' \subseteq S$, $\Ad$ also wins the game starting from $(v,S')$. Moreover, if $\tau$ is a winning strategy for $\Ad$ from $(v,S)$, then the strategy $\tau'$ where $\tau'(v',S' \cup U) = \tau(v',S \cup U)$ for all $v' \in V$ and $U \subseteq [k]$ is also a winning strategy of $\Ad$ starting from $(v,S')$.
\end{lem}
\begin{proof}
Let $\sigma$ be a strategy of $\Ev$ in $\game \times \M$.
Consider a play $\pi$ in $\game \times \M$ from $(v,S')$ consistent with $\tau'$ and $\sigma$. 
Let $\pi'$ be the corresponding play in $\game \times \M$ from $(v,S)$ consistent with $\tau$ and $\sigma$. Since $\tau'(v',S' \cup U) = \tau(v',S \cup U)$, if at some point the play $\pi$ reaches some position $(u,S' \cup U)$ then the play $\pi'$ reaches the position $(u,S \cup U)$.

If $\pi$ is losing for $\Ad$, i.e. it reaches some position $(u,[k])$, then
$S' \cup U = [k]$, which implies $S \cup U = [k]$.
But then $\pi'$ also reaches $(u,[k])$,
contradicting that $\tau$ is winning for $\Ad$ from $(v,S)$. Hence $\tau'$ is winning for $\Ad$ from $(v,S')$.
\end{proof}

We note that a property similar to Lemma~\ref{lem:adam-mem-downward-closed} does not hold for $\Ev$. Consider the game in Figure~\ref{fig:optimisme}, where $\Ev$ reaches the targets $\F = \{ \{1\}, \{2\} \}$ from $(v_0,\{1\})$ by going right from $v_0$. However, to win from $(v_0,\emptyset)$ the same strategy does not work; she has to go left first to visit target $1$ before going right to visit target $2$. Moreover, just subset inclusion of memory states may not preserve winning for $\Ev$. $\Ev$ trivially wins from $(2, \{1,2\})$ but she cannot win starting from $(2, \{2\})$.
\begin{figure}[ht]
\begin{center}
\begin{tikzpicture}[shorten >=0.7pt, on grid, auto, thick, >=stealth']

\tikzstyle{state}=[circle, draw, fill=gray!10, minimum size=8mm]
\tikzstyle{box}=[rectangle, draw, fill=gray!10, minimum size=8mm]
	\node[state] (coeur) at (1.6,0) {$v_0$};
	\node[box,double] (1) at (0,0) {$1$};
  	\node[box,double] (2) at (3.2,0) {$2$};

	\draw[->] (coeur) to[bend right=22] (1);
	\draw[->] (coeur) -- (2);
	\draw[->] (1) to[bend right=22] (coeur);
	\draw[->] (2) edge[loop right] ();
\end{tikzpicture}
\end{center}
\caption{A game where Lemma~\ref{lem:adam-mem-downward-closed} fails for $\Ev$ }
\label{fig:optimisme}
\end{figure}
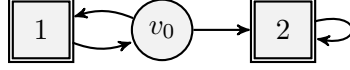

Now we describe our results for memory bounds. We begin with upper bounds.

\begin{thm}[Upper bounds on memory]
\label{thm:memupper}
For every generalised reachability game $\game = (\Aa,s,\F)$ with $\F=\{F_1,\ldots,F_k\}$,
\begin{enumerate}
	\item if $\Ev$ wins, then she has a winning strategy with $2^k - 1$ memory states;
	\item if $\Ad$ wins, then he has a winning strategy with $\binomksurdeuxk$ memory states.
\end{enumerate}
\end{thm}

\begin{proof}
From Lemma~\ref{lem:product-game-equivalence} it follows that the winners of $\game$ and that of $\game \times \M$ with reachability objective of reaching $V \times \{[k]\}$ are the same. Since reachability games can be won with memoryless strategies, it follows that in $\game$ each player has a winning strategy using memory $\M$.
Thus both players have a winning strategy using $2^k$ memory states.
We will proceed to make the bounds more precise.
\begin{enumerate}
	\item $\Ev$ does not need a special memory state to remember that all target sets
have been visited, since at that point she has already won.
Therefore, she can always win with $2^k-1$ memory states:
the memory structure $\M$ is the same, except that the state $\set{1,\ldots,k}$ is no longer needed,
and in the update function, $\set{1,\ldots,k}$ is replaced by any state,
which does not matter since the play is already winning. 
\item For $\Ad$, Lemma~\ref{lem:adam-mem-downward-closed} allows us to merge memory states:

For $v \in V$, consider the set $W_v$ of subsets $S$ such that $(v,S)$ is winning for $\Ad$, using some winning strategy $\tau$.
The maximal elements of $W_v$ (for the relation set inclusion) are pairwise incomparable, hence there are at most $\binomksurdeuxk$ of them;
denote them by $S_1(v),\ldots,S_p(v)$.
Hence, at any node, $\Ad$ only needs to consider at most $p(v)$ memory states, where $p(v) \le \binomksurdeuxk$.
%The idea is the following: from $v$, $\Ad$ only considers one of the $p$ ``worst cases'', namely one of $S_1(v),\ldots,S_p(v)$.

Define a memory structure $\M' = (M',i_0,\upd')$.
where $M' \subseteq 2^{[k]}$ are all the maximal elements of $W_v$ for nodes $v$. %We have $|M'| \le \binomksurdeuxk$.
Since $(s,m_0)$ is winning for $\Ad$ in $\game \times \M$, there exists $i_0 \in M'$ such that $i_0$ contains $m_0$;
we take such an $i_0$ as the initial memory state.
We define the update function by letting $\upd'(i,(v,v'))$ be a $j \in M'$ such that
$j$ contains $\upd(i,(v,v'))$, if such a $j$ exists.

We now define the strategy $\tau'$ that chooses the next move from $(v,i)$.
For every $(v,S)$ winning for $\Ad$ in $\game \times \M$ using $\tau$, there exists a choice $(v',\upd(S,(v,v')))$ using $\tau$ that is also winning for $\Ad$.
Applied to $(v,i)$, this guarantees the existence of a vertex $v'$: we define $\tau'(v,i) = v'$.
For every finite play $w$ ending in $v \in V$ in which $\Ad$ plays according to $\tau'$, we show that the resulting memory state (\textit{i.e.} $\upd'^+(w)$) is an $i$ such that $i$ contains the set of visited target sets,
\textit{i.e.} $\upd^+(w) \subseteq i$.
We show this by induction. It holds initially by construction,
since $m_0 = \upd^+(s) \subseteq i_0$.
Assume that the last transition is $(v,v') \in E$.
Since $\Ad$ wins from $(v,i)$, he also wins from $(v',\upd(i,(v,v')))$:
either because $\Ev$ cannot leave $\Ad$'s winning region,
or by construction of $\Ad$'s strategy.
Therefore there exists a $j \in M'$ such that $j$ contains $\upd(i,(v,v'))$,
in other words $\upd'(i,(v,v'))$ is well defined.
By the induction hypothesis, $\upd^+(w) \subseteq i$,
hence $\upd^+(w \cdot v') \subseteq \upd(i,(v,v')) \subseteq j$.
This invariant implies that the play always remains in $\Ad$'s winning region,
and is therefore winning for $\Ad$.
The number of memory states used at any instance is $\binomksurdeuxk$.
\end{enumerate}
\hfill\end{proof}
The bound on the memory for $\Ad$ essentially comes from the size of the largest antichain in the set of subsets of $[k]$, and this applies to all topologically closed objectives~\cite{ColcombetFH14}.

We now prove that the upper bounds above are tight, by providing matching lower bounds.

\begin{thm}[Lower bounds on memory]
\label{thm:memlower}	
For every $k \ge 1$,
\begin{enumerate}
	\item there exists a generalised reachability game $\game = (\Aa,s,\F)$ with $\F=\{F_1,\ldots,F_k\}$
	in which $\Ev$ needs $2^k - 1$ memory states to win;
	\item there exists a generalised reachability game $\game = (\Aa,s,\F)$ with $\F=\{F_1,\ldots,F_k\}$
	in which $\Ad$ needs $\binomksurdeuxk$ memory states to win.
\end{enumerate}
\end{thm}
\begin{proof}

\begin{enumerate}
	\item We first describe a generalised reachability game in which $\Ev$ needs $2^k - 1$ memory states to win.
The arena is illustrated in Figure~\ref{fig:flower}, for $k = 5$, and appears in~\cite{ChatterjeeHH11}.

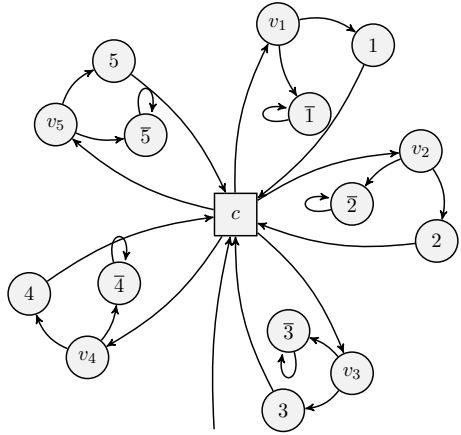
\begin{figure}[!ht]
\begin{center}
\scalebox{0.7}{
\begin{tikzpicture}[on grid, auto, thick, >=stealth']
\tikzstyle{box}=[circle, draw, fill=gray!10, minimum size=8mm]
\tikzstyle{state}=[rectangle, draw, fill=gray!10, minimum size=8mm]

	\node[state] (coeur) at (0,0) {$c$};
	\node[draw=none] (racine) at (-0.4,-4.2) {};
	\draw[->] (racine) to[bend left=8] (coeur);

	\node[box] (a1) at (0.8,3.6) {$v_1$};
	\node[box] (b1) at (2.6,3.2) {$1$};
	\node[box] (c1) at (1.4,1.9) {$\overline{1}$};

	\draw[->] (coeur) to[bend left=14] (a1);
	\draw[->] (a1) to[bend left=24] (b1);
	\draw[->] (b1) to[bend left=14] (coeur);
	\draw[->] (a1) to[bend right=16] (c1);
	\draw[->] (c1) edge[loop left] ();

	\node[box] (a2) at (3.5,1.2) {$v_2$};
	\node[box] (b2) at (3.8,-0.5) {$2$};
	\node[box] (c2) at (2.2,0.2) {$\overline{2}$};

	\draw[->] (coeur) to[bend left=14] (a2);
	\draw[->] (a2) to[bend left=24] (b2);
	\draw[->] (b2) to[bend left=14] (coeur);
	\draw[->] (a2) to[bend right=16] (c2);
	\draw[->] (c2) edge[loop left] ();

	\node[box] (a3) at (2.2,-3.0) {$v_3$};
	\node[box] (b3) at (0.9,-3.7) {$3$};
	\node[box] (c3) at (1.0,-2.2) {$\overline{3}$};

	\draw[->] (coeur) to[bend left=14] (a3);
	\draw[->] (a3) to[bend left=24] (b3);
	\draw[->] (b3) to[bend left=14] (coeur);
	\draw[->] (a3) to[bend right=16] (c3);
	\draw[->] (c3) edge[loop below] ();

	\node[box] (ai) at (-2.8,-2.7) {$v_4$};
	\node[box] (bi) at (-3.9,-1.5) {$4$};
	\node[box] (ci) at (-2.2,-1.3) {$\overline{4}$};

	\draw[->] (coeur) to[bend left=14] (ai);
	\draw[->] (ai) to[bend left=24] (bi);
	\draw[->] (bi) to[bend left=14] (coeur);
	\draw[->] (ai) to[bend right=16] (ci);
	\draw[->] (ci) edge[loop above] ();

	\node[box] (ak) at (-3.4,1.7) {$v_5$};
	\node[box] (bk) at (-2.3,2.9) {$5$};
	\node[box] (ck) at (-1.7,1.5) {$\overline{5}$};

	\draw[->] (coeur) to[bend left=14] (ak);
	\draw[->] (ak) to[bend left=24] (bk);
	\draw[->] (bk) to[bend left=14] (coeur);
	\draw[->] (ak) to[bend right=16] (ck);
	\draw[->] (ck) edge[loop above] ();
\end{tikzpicture}
}
\end{center}
\caption{A generalised reachability game in which $\Ev$ needs $2^k - 1$ memory states to win.}
\label{fig:flower}
\end{figure}

The arena has a flower shape : vertex $c$ controlled by $\Ad$ at the heart of the flower; the flower has $k$ petals where each petal consists of 3 vertices $v_i$, $i$, and $\overline{i}$, all controlled by $\Ev$. There are $k$ target sets $F_1, \ldots, F_k$, where $F_i = \{i\} \cup \{\overline{j} | j \neq i\}$. % A vertex labelled by $i$ belongs to $F_i$, and a vertex labelled by $\overline{i}$ belongs to all $F_j$ except $F_i$.
A play starts at the heart $c$ of the flower;
first $\Ad$ chooses a petal $i$,
then from $v_i$ $\Ev$ either chooses to visit target $i$ before returning to the heart
(in which case the play continues), or to visit all targets except $i$ and stop the play.
$\Ev$ wins with the following strategy: the first time $\Ad$ chooses petal $i$,
$\Ev$ returns to the heart; the second time $\Ad$ chooses petal $i$,
$\Ev$ stops the play.
This strategy is easily implemented with $2^k$ memory states.
One of these states is in fact unnecessary: it is the one corresponding to the case where all petals have been seen once,
at which point $\Ev$ has already won.

We show that there is no winning strategy for $\Ev$ using fewer than $2^k - 1$ memory states.
Let $\sigma$ be a strategy using memory $\M$ with fewer than $2^k - 1$ memory states, defined by the function $\next$.
For each memory state $m$, consider
$S_m = \set{i \mid \next(v_i,\upd(m,(c,v_i))) = \overline{i}}$,
the set of petals from which $\Ev$ would choose to stop the game if $\Ad$ proposed that petal.
Since there are fewer than $2^k - 1$ memory states, there exists a strict subset $X$ of $\set{1,\ldots,k}$
that is not of the form $S_m$ for any memory state $m$.
$\Ad$ wins against strategy $\sigma$ by choosing at each step a petal
in the symmetric difference of $X$ and $S_m$, where $m$ is $\Ev$ current memory state.
Indeed:
\begin{itemize}
	\item if $\Ad$ keeps choosing petals from $X$ forever, then $\Ev$ never stops the play and only the targets
from $X$ are visited;
	\item otherwise, when $\Ev$ stops the play, the last time the token was on the heart,
$\Ev$'s memory state was some $m$ such that $X \subsetneq S_m$,
and the chosen petal is an $i$ belonging to $S_m \setminus X$, so target $i$ has never been visited.
\end{itemize}

\item We now describe a generalised reachability game in which $\Ad$ needs $\binomksurdeuxk$ memory states to win.
Let $k = 2p + 1$. The game consists of 3 similar gadgets placed sequentially, where the first and last gadgets are controlled by $\Ev$ and the middle one is controlled by $\Ad$. 
Figure~\ref{fig:lower_bound_adam} illustrates a gadget for $p=3$, where vertices are labelled with the targets and one player (here, $\Ev$) chooses exactly $p$ targets out of $2p+1$ in one pass of the gadget.
\begin{figure}[!ht]
\begin{center}
\scalebox{0.7}{
\begin{tikzpicture}[shorten >=0.7pt, on grid, auto, thick, >=stealth']

\tikzstyle{state}=[circle, draw, fill=gray!10, minimum size=5mm]
\tikzstyle{box}=[circle, draw, fill=gray!10, minimum size=5mm]

	\node[state] (init) at (0,0) {$v_0$};

	\node[box] (15) at (2,2.4) {$5$};
	\node[box] (14) at (2,1.2) {$4$};
	\node[box] (13) at (2,0) {$3$};
	\node[box] (12) at (2,-1.2) {$2$};
	\node[box] (11) at (2,-2.4) {$1$};

	\node[box] (26) at (4.6,3.0) {$6$};
	\node[box] (25) at (4.6,1.8) {$5$};
	\node[box] (24) at (4.6,0.6) {$4$};
	\node[box] (23) at (4.6,-0.6) {$3$};
	\node[box] (22) at (4.6,-1.8) {$2$};

	\node[box] (37) at (7.2,3.6) {$7$};
	\node[box] (36) at (7.2,2.4) {$6$};
	\node[box] (35) at (7.2,1.2) {$5$};
	\node[box] (34) at (7.2,0) {$4$};
	\node[box] (33) at (7.2,-1.2) {$3$};

	\node[state] (final) at (9.8,0) {$v_1$};

	\draw[->] (init) -- (11);
	\draw[->] (init) -- (12);
	\draw[->] (init) -- (13);
	\draw[->] (init) -- (14);
	\draw[->] (init) -- (15);

	\draw[->] (11) -- (22);
	\draw[->] (11) -- (23);
	\draw[->] (11) -- (24);
	\draw[->] (11) -- (25);
	\draw[->] (11) -- (26);

	\draw[->] (12) -- (23);
	\draw[->] (12) -- (24);
	\draw[->] (12) -- (25);
	\draw[->] (12) -- (26);

	\draw[->] (13) -- (24);
	\draw[->] (13) -- (25);
	\draw[->] (13) -- (26);

	\draw[->] (14) -- (25);
	\draw[->] (14) -- (26);

	\draw[->] (15) -- (26);

	\draw[->] (22) -- (33);
	\draw[->] (22) -- (34);
	\draw[->] (22) -- (35);
	\draw[->] (22) -- (36);
	\draw[->] (22) -- (37);

	\draw[->] (23) -- (34);
	\draw[->] (23) -- (35);
	\draw[->] (23) -- (36);
	\draw[->] (23) -- (37);

	\draw[->] (24) -- (35);
	\draw[->] (24) -- (36);
	\draw[->] (24) -- (37);

	\draw[->] (25) -- (36);
	\draw[->] (25) -- (37);

	\draw[->] (26) -- (37);

	\draw[->] (33) -- (final);
	\draw[->] (34) -- (final);
	\draw[->] (35) -- (final);
	\draw[->] (36) -- (final);
	\draw[->] (37) -- (final);
\end{tikzpicture}
}
\end{center}
\caption{A gadget where $\Ev$ chooses exactly $p$ targets out of $2p+1$.}
\label{fig:lower_bound_adam}
\end{figure}
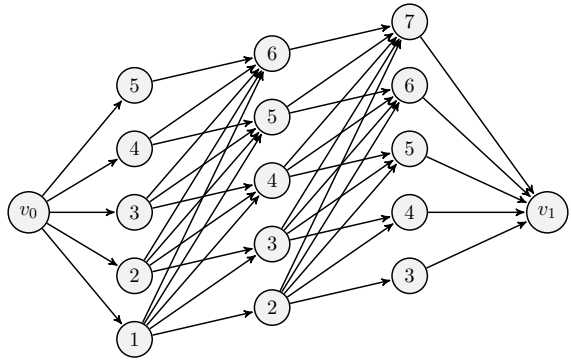
The whole game arena is comprised of the gadget being repeated three times:
\begin{enumerate}
	\item In first copy of gadget $\Ev$ chooses $p$ targets,
	\item In second copy of gadget $\Ad$ chooses $p$ targets,
	\item In third and final copy of gadget $\Ev$ chooses $p$ targets.
\end{enumerate}
To win, in second copy of gadget $\Ad$ must visit exactly the same targets that $\Ev$ visited in the first copy,
which requires $\binomksurdeuxk$ memory states;
otherwise at least $p+1$ targets have been visited when the third copy begins,
and $\Ev$ wins by choosing the targets that have not been visited, of which there are at most $p$.
\end{enumerate}
\hfill\end{proof}

We do not know the exact complexity of generalised reachability games where the target sets have size $2$.
In the remaining of this subsection, we discuss this question, focusing on memory requirements for both players.
The memory required for Eve is still exponential, as shown in Figure~\ref{fig:counter_ex_memory} for $k = 4$.
Specifically, it shows a generalised reachability game where reachability sets have size $2$ won by Eve, 
where she needs $2^{\ksurdeux + 1} - 1$ bits of memory to win. On the other hand, the exact memory required for Adam remains open.
Figure~\ref{fig:adamsize2}~%, following an idea of Christof Loeding,
%\red{need reference, paper or personal communication?},
 shows a generalised reachability game where reachability sets have size $2$ won by Adam, 
where he needs $4$ memory states to win.
\begin{thm}
For every $k \ge 1$,
\begin{enumerate}
	\item there exists a family of generalised reachability games $\game = (\Aa,s,\F)$ with $\F=\{F_1,\ldots,F_k\}$ 
    such that $|F_i|\leq 2$ for all $1\leq i\leq k$,
    in which $\Ev$ needs $2^{\ksurdeux + 1} - 1$ memory states to win;
	\item there exists a generalised reachability game $\game = (\Aa,s,\F)$ with $\F=\{F_1,\ldots,F_k\}$ 
    such that $|F_i|\leq 2$ for all $1\leq i\leq k$,
	in which $\Ad$ needs $4$ memory states to win.
\end{enumerate}
\end{thm}
\begin{proof}

First we present the family of games for $\Ev$. An example of such a game is shown in Figure~\ref{fig:counter_ex_memory} for $k = 4$, where vertices in the same target set has the same label.
The arena is divided into two parts: the left hand side is a flower with $\ksurdeux$ petals,
and the right hand side a one-player arena controlled fully by $\Ev$.
The game starts at the heart $c$ of the flower and $\Ev$ has the option of visiting one of the $\ksurdeux$ petals or leaving the flower to go to the right hand side.
%Once this task is completed, she can move to the right hand side to reach the remaining target sets.
There are $k$ target sets, $1 \dots k$, each of size $2$, and each petal contains exactly one vertex from two different target sets, i.e. petal $i$ contains target $2i-1, 2i$ for $1 \leq i \leq \ksurdeux$.
$\Ev$ needs to remember the $\ksurdeux$ choices made by $\Ad$ (one for each petal), in order to reverse them when visiting right hand side of the arena. For example:
if $\Ad$ chose the target set $1$, then the target set $2$ has not been reached,
so $\Ev$ has to choose target set $2$ later.
Remembering those choices and asking for each petal requires $2^{\ksurdeux + 1} - 1$ memory states.
(Note that, this is the size of the complete binary tree of depth $\ksurdeux$.)

\begin{figure}[ht]
\begin{center}
\scalebox{0.7}{
\begin{tikzpicture}[on grid, auto, thick, >=stealth']
\tikzstyle{state}=[circle, draw, fill=gray!10, minimum size=8mm]
\tikzstyle{box}=[rectangle, draw, fill=gray!10, minimum width=8mm, minimum height=6mm]

	\node[state] (n0) at (0,0) {$c$};
	\draw[->] (-1.0,0) -- (n0);

	\node[box] (p1) at (0,2.6) {};
	\node[state] (n1) at (-1.2,1.3) {$1$};
	\node[state] (n2) at (1.2,1.3) {$2$};

	\draw[->] (n0) -- (p1);
	\draw[->] (p1) to[bend left=18] (n2);
	\draw[->] (p1) to[bend right=18] (n1);
	\draw[->] (n1) to[bend right=18] (n0);
	\draw[->] (n2) to[bend left=18] (n0);

	\node[box] (p2) at (0,-2.6) {};
	\node[state] (n3) at (-1.2,-1.3) {$3$};
	\node[state] (n4) at (1.2,-1.3) {$4$};

	\draw[->] (n0) -- (p2);
	\draw[->] (p2) to[bend left=18] (n3);
	\draw[->] (p2) to[bend right=18] (n4);
	\draw[->] (n3) to[bend left=18] (n0);
	\draw[->] (n4) to[bend right=18] (n0);

	\node[state] (c1) at (3.2,0) {};
	\node[state] (d1) at (4.6,-1.1) {$1$};
	\node[state] (d2) at (4.6,1.1) {$2$};
	\node[state] (c2) at (6.0,0) {};
	\node[state] (d3) at (7.4,-1.1) {$3$};
	\node[state] (d4) at (7.4,1.1) {$4$};

	\draw[->] (n0) -- (c1);
	\draw[->] (c1) -- (d1);
	\draw[->] (c1) -- (d2);
	\draw[->] (d1) -- (c2);
	\draw[->] (d2) -- (c2);
	\draw[->] (c2) -- (d3);
	\draw[->] (c2) -- (d4);
	\draw[->] (d3) edge[loop right] ();
	\draw[->] (d4) edge[loop right] ();
\end{tikzpicture}
}
\end{center}
\caption{A generalised reachability game where $\Ev$ needs $2^{\ksurdeux + 1} - 1$ memory states to win.}
\label{fig:counter_ex_memory}
\end{figure}
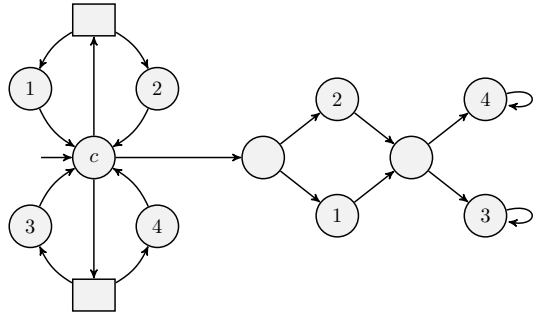

Now we describe the game for $\Ad$ as show in \Cref{fig:adamsize2}, where he needs $4$ memory states to win. The game has $4$ targer sets, $1 \dots 4$ each of size $2$, and vertices in the same target set has the same label. Label $\neg i$ denotes that the vertex belongs to all target sets except $i$.
The game starts from vertex $s$. First $\Ev$ chooses and isits three of the four target sets 
(two targets in the first choice, $1$ and $2$, or $3$ and $4$, and then one in the second choice),
and move to vertex $t$, controlled by $\Ad$.
There, $\Ad$ has four options, each allowing all target sets but one.
Remembering the four possibilities requires four memory states, and leads to a win.
However, with less memory states, one of the four options will never be played, and $\Ev$ wins.

\begin{figure}[ht]
\begin{center}
\scalebox{0.6}{
\begin{tikzpicture}[
    on grid,
    auto,
    thick,
    >=stealth',
    shorten >=1pt,
    shorten <=1pt
]
\tikzstyle{state}=[circle, draw, fill=gray!10, minimum size=8mm]
\tikzstyle{box}=[rectangle, draw, fill=gray!10, minimum width=8mm, minimum height=8mm]

	\node[state] (0) at (0,0) {$s$};
	\draw[->] (-1.2,0) -- (0);

	\node[state] (1)  at (2.0, 1.8) {$1$};
	\node[state] (2)  at (4.0, 1.8) {$2$};
	\node[state] (3)  at (2.0,-1.8) {$3$};
	\node[state] (4)  at (4.0,-1.8) {$4$};

	\node[state] (4b) at (6.5, 3.0) {$4$};
	\node[state] (3b) at (6.5, 1.0) {$3$};
	\node[state] (2b) at (6.5,-1.0) {$2$};
	\node[state] (1b) at (6.5,-3.0) {$1$};

	\node[state] (234) at (10.3, 3.0) {$\neg 1$};
	\node[state] (134) at (10.3, 1.0) {$\neg 2$};
	\node[state] (124) at (10.3,-1.0) {$\neg 3$};
	\node[state] (123) at (10.3,-3.0) {$\neg 4$};

	\node[box] (c) at (13.4,0) {$t$};

	\draw[->] (0) -- (1);
	\draw[->] (0) -- (3);

	\draw[->] (1) -- (2);
	\draw[->] (3) -- (4);

	\draw[->] (2) -- (4b);
	\draw[->] (2) -- (3b);
	\draw[->] (4) -- (2b);
	\draw[->] (4) -- (1b);

	% route these around the negated nodes
	\draw[->] (4b) .. controls (8.2,4.5) and (11.7,4.5) .. (c);
	\draw[->] (3b) .. controls (8.4,2.2) and (11.8,2.2) .. (c);
	\draw[->] (2b) .. controls (8.4,-2.2) and (11.8,-2.2) .. (c);
	\draw[->] (1b) .. controls (8.2,-4.5) and (11.7,-4.5) .. (c);

	\draw[->] (c) -- (123);
	\draw[->] (c) -- (124);
	\draw[->] (c) -- (134);
	\draw[->] (c) -- (234);

	\draw[->] (123) -- (1b);
	\draw[->] (123) -- (2b);
	\draw[->] (123) -- (3b);

	\draw[->] (124) -- (1b);
	\draw[->] (124) -- (2b);
	\draw[->] (124) -- (4b);

	\draw[->] (134) -- (1b);
	\draw[->] (134) -- (3b);
	\draw[->] (134) -- (4b);

	\draw[->] (234) -- (2b);
	\draw[->] (234) -- (3b);
	\draw[->] (234) -- (4b);
\end{tikzpicture}
}
\end{center}
\caption{A generalised reachability game where $\Ad$ needs $4$ memory states to win.}
\label{fig:adamsize2}
\end{figure}
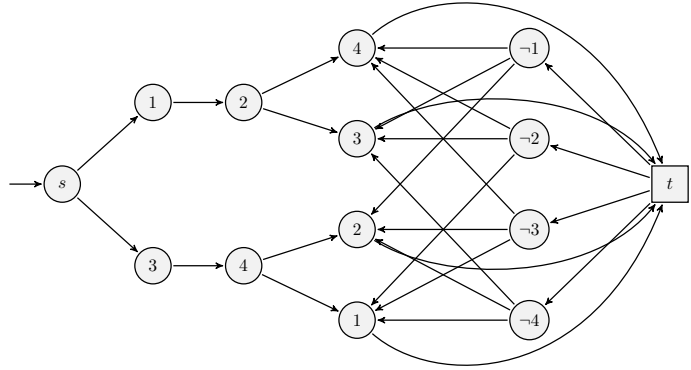
\end{proof}

Quite surprisingly, we could not generalize this example to obtain a better lower bound than $4$.
We do not know whether this bound is tight (in any arena, if $\Ad$ wins, then he has a winning strategy with $4$ memory states),
which is plausible. 
Note that this would imply a \coNP$^{\textsc{NP}}$ algorithm: guess a winning strategy for $\Ad$ with $4$ memory states, 
and compose this strategy with the game, then solve the resulting one-player game.

\section{Maximum Generalised Reachability}

In this section, we address the complexity of the optimisation problems $\maxgenreach$ and the $\maxgenreachpromise$. %
\subsection{One Player Case with $\Ev$}
We first consider the case where all vertices of the game are controlled by $\Ev$. In this case, a strategy of $\Ev$ produces a unique play $\rho$. Therefore, the $\maxgenreach$ and $\maxgenreachpromise$ problems are equivalent if $\Ev$ pledges the set of targets seen in the play $\rho$. Hence we only state our results for the $\maxgenreach$ variant.

\begin{thm}
\label{thm:maxgrev}
Let $\game$ be a game where $\Ev$ controls all the vertices, i.e, $V=V_{\Ev}$.
\begin{enumerate}
    \item $\maxgenreach$ is in \P~when each target set $F_i \in \F$ is of size $1$.
    \item $\maxgenreach$ is \NP-complete in general. It is \NP-hard even when $|F_i| = 2$ for each target set $F_i \in \F$. 
\end{enumerate}
\end{thm}

\begin{proof}
For the case with all target sets of size $1$, we present an algorithm that runs in polynomial time. The algorithm first computes the strongly connected component (SCC) decomposition of the arena. For each SCC, we assign a value equal to the number of target vertices contained in it. Within each SCC, $\Ev$ can visit all the target vertices. Computing the maximum number of target vertices that $\Ev$ can visit corresponds to finding a path in the SCC decomposition that maximises the sum of the value of the SCCs contained in the path.
This can be computed bottom-up using dynamic programming starting from the bottom SCCs. 
This is explained in detail in the proof of Theorem \ref{thm:maxgrp2p}, as the algorithm also works for the two-player case.
When target sets are of size $2$, \NP-hardness follows from MAX-2-SAT problem, which is known to be \NP-hard~\cite{Halaby2016}. This follows the same reduction as the one used to show \PSPACE-hardness for the general case with 2 players and arbitrary target sets in Theorem~\ref{thm:pspace-hardness-general-case}.

To see membership in NP, $\Ev$ can guess a path of size at most $nk$ in length that visits $k$ target sets, where $n$ is the number of vertices in the arena. This is possible as the length of the path between any two consecutive target vertices not seen before is at most $n$.
\end{proof}

\subsection{One Player Case with $\Ad$}

For one-player games with $\Ad$, for both the $\maxgenreach$ and $\maxgenreachpromise$ problems, $\Ad$ always has an optimal strategy that produces a lasso. This follows from the fact that it is always better for $\Ad$ to encounter fewer distinct vertices, and consequently fewer target sets. 

Recall that the $\maxgenreach$ and $\maxgenreachpromise$ problems can have different solutions even when $\Ad$ is the sole player, as demonstrated by a game obtained by modifying Figure~\ref{fig:maxproblem}, where all vertices belong to Adam : in this case, Eve cannot promise to visit any \emph{particular} target vertex, but she can ensure that at least one target vertex is visited.
In this section we deal with the one-player variant of these problems with $\Ad$ as the only player. First we provide the full complexity picture for the $\maxgenreach$ problem. 
\begin{thm}
\label{thm:maxgrad}
Let $\game$ be a game where $\Ad$ controls all the vertices, i.e, $V=V_{\Ad}$. 
\begin{enumerate}
    \item  $\maxgenreach$ is in $\P$ when all target sets are singleton.
    \item $\maxgenreach$ is $\coNP$-complete in general. The $\coNP$-hardness holds even when $|F_i| = 2$ for each $F_i \in \F$. 
\end{enumerate}
\end{thm}
\begin{proof}
We begin with a polynomial time algorithm for the case where target sets have size $1$. Note that a strategy of $\Ad$ is just a path $\rho$. We can assume that $\rho$ is a path of the form $s\to t\to t$. Otherwise, one can find a minimal loop in $\rho$ and simply repeat it to obtain a path $\rho'$ of the correct form. Since $\rho'$ visits a subset of vertices visited by $\rho$, $\rho'$ should be at least as good as $\rho$ for $\Ad$. Therefore, we search for paths of the form $s\to t\to t$ which visit the least number of target vertices.

We assign weights to edges of the graph $G$ using the function $w: E\to \{0,1\}$ as follows: 
$w(u,v)=1$ if, and only if, $\{v\}\in \F$.
The problem reduces to computing the lightest weighted $s\to t\to t$ path in the weighted graph. 
This can be done by iterating over all choices of $t$ and computing the lightest such path.

For co-NP hardness, we look at the complement problem, i.e. given a game $\game$ with $V=V_\Ad$ and $k$, is there a strategy of $\Ad$ such that on playing this strategy at most $k$ targets are visited. We consider the problem MIN-2-SAT, which is known to be NP-hard \cite{Kohli-et-al-94}. 
Following the same reduction as in  NP-hardness in proof of Theorem \ref{thm:maxgrev}, but giving control of the vertices to $\Ad$, we get a game in which $\Ad$ can visit at most $k$ target vertices if and only if there is an assignment which satisfies at most $k$ clauses.  

For the upper bound, we observe that $\Ad$'s strategy can be simplified to a path of length at most $n+1$, because once he sees a vertex twice, he can simply repeat the loop without visiting more target sets. We can thus guess a path of length $n+1$ and check if this hits at most $k$ targets. This puts the complement problem in NP and thus shows the co-NP membership. 
\end{proof}

Note that the algorithm to solve $\genreach$ games where all vertices are controlled by $\Ad$ computes $\AtE$ for each target set, and checks if the initial vertex is in the $\AtE$ for all target sets or not. If not, $\Ad$ has a choice to avoid a target set and win. This runs in polynomial time as computing the 
$\AtE$ can be done in polynomial time. However, this does not work for the $\maxgenreach$ problem as the initial state $s$ might not be in the $\AtE(F_i)$ and $\AtE(F_j)$, but in the $\AtE(F_i\cup F_j)$. Thus, the play will visit at least $1$ of the target sets.
For example, in the game from Figure~\ref{fig:maxproblem}, $s$ and $u$ are neither in $\AtE(\{u_1\})$, nor in $\AtE(\{u_3\})$, but they are in $\AtE(\{u_1,u_3\})$.

Next, we show that the polynomial time algorithm can be adapted to work for the $\maxgenreachpromise$ problem.

\begin{thm}
\label{thm:maxgrprad}
Let $\game$ be a game where $\Ad$ controls all the vertices, i.e, $V=V_{\Ad}$. Then,~\\   
  $\maxgenreachpromise$  is in \P.
\end{thm}
\begin{proof}
    The $\maxgenreachpromise$ problems asks, if there is a $k$ sized subset of target sets such that the play visits a target from each of the $k$ chosen target sets. 
    The algorithm computes the $\AtE(F_i)$ for all target sets $F_i$ and counts how many of them contain the initial state $s_0$. If $s_0$ is contained in at least $k$ of the $\AtE$ sets, then any play will visit the $k$ target sets. This is because being in the attractor set of a target $F_i$ guarantees that Eve has a strategy to force the play to reach $F_i$, even though Adam controls all the vertices. Therefore, $\Ev$ can promise $k$ target sets whose $\AtE$ contains $s_0$. If $s_0$ is not contained in $k$ attractor sets, then no matter which $k$ sets $\Ev$ promises, there will be a promised target set $F'$, such that $s_0\not \in \AtE(F)$. Therefore, $\Ad$ will be able to avoid the promised target sets.
\end{proof}

\subsection{Two Player Case}
Here we discuss the complexity of the two-player case.
We first state the results for the $\maxgenreach$ problem. %

\begin{thm}\label{thm:maxgrcomplexity}
Let $\game$ be a game.
\begin{enumerate}
\item $\maxgenreach$ is \PSPACE-complete. The \PSPACE-hardness holds even when $|F_i| \leq 2$ for each $F_i \in \F$.
\item $\maxgenreach$ is \coNP-hard when each $F_i \in \F$ is singleton.

\end{enumerate}
\end{thm}

\begin{proof}
(1)     
The upper bound in the general case, i.e, $|F_i|\ge 2$, for all $F_i\in \F$, follows from \cite[Theorem 12]{KupfermanS24}. In fact, they consider the question of maximizing weighted reachability, where different reachability objectives can have different weight associated to them and $\Ev$ attempts to maximise the total weight. Our case corresponds to the case where all the reachability objectives have the same weight, which is called the \textit{MaxR} objective in \cite{KupfermanS24}.

%The PSPACE lower bound presented in \cite{KupfermanS24} follows from the PSPACE-hardness of solving generalised reachability games, which works for target sets of size $\ge 3$. 
We now present the proof for \PSPACE lower bound even for the case where target sets are of size $2$.
We use the fact that MAX-2-SAT is \NP-hard~\cite{GAREY1976237} as well the reduction used for this.
In the reduction from 3-SAT to MAX-2-SAT~\cite{GAREY1976237}, given a 3-SAT instance $\phi$ over $n$ variables, $X = \{x_1, \dots, x_n\}$, and $m$ clauses, $C_1,\dots,C_m$, they construct a 2-SAT instance $\phi'$ with $n+m$ variables and $10m$ clauses in the following way: for each clause $C_i = (l_i^1 \vee l_i^2 \vee l_i^3)$ in $\phi$ a new variable $y_i$ is introduced, and there are $10$ clauses $C^1_i,\dots, C^{10}_i$ in $\phi'$ corresponding to $C_i$. It satisfies the following crucial property : for any clause $C_i$ and any assignment of the variables $X$, $C_i$ is satisfied iff there is an assignment of $y_i$ under the same assignment of the variables $X$, such that exactly $7$ of the clauses $C^1_i,\dots, C^{10}_i$ are satisfied. If $C_i$ is not satisfied, then no more than $6$ clauses among $C^1_i,\dots, C^{10}_i$ are satisfied irrespective of the value of $y_i$.  %They create a set of 2-CNF clauses such that 
Hence exactly $\frac{7}{10}$ of the 2-CNF clauses are satisfied if the original 3-CNF clause was satisfied. Otherwise, at most $\frac{6}{10}$ of the 2-CNF clauses are satisfied. 

Our reduction follows the same construction turning a 3-QBF in CNF instance into a 2-QBF instance. New variables $y_i$'s are made innermost existential variables, i.e. at the innermost quantifier level. %Since this holds for any assignment, we can transform the matrix of the QBF formula from a 3-CNF to a 2-CNF formula such that exactly $\frac{7}{10}$ of the clauses are satisfied by a satisfying assignment. 
Since $y_i$'s are existential variables, for each clause $C_i$ and for each assignment of $X$, the same property holds for the $10$ clauses corresponding to $C_i$ in the 2-QBF instance. Hence the original QBF formula is true iff there is an assignment of the variables $y_i$'s such that for any assignment of the variables $X$, at least $\frac{7}{10}$ of the clauses in the 2-QBF instance are satisfied. 
%Therefore, this proves that MAX-2-QSAT is PSPACE-hard. 
Now we use the same construction to translate the 2-QBF formula to a generalised reachability game, as in the proof of Theorem \ref{thm:pspace-hardness-general-case}, where each clause corresponds to a target set of size $2$. The original 3-QBF formula is true iff $\Ev$ has a strategy to visit at least $\frac{7}{10}$ of the target sets. Hence, the $\maxgenreach$ problem with target sets of size $2$ is \PSPACE-hard.
%Using the construction to translate QSAT to generalised reachability games, we obtain that $\maxgenreach$ with target sets of size $2$ is PSPACE-hard.

(2) 
We reduce the complement of $\maxgenreach$ from the minimum vertex cover problem, which is known to be NP-complete~\cite{Karp72}.
Given an undirected graph $G = (V,E)$ and a number $\ell$, the minimum vertex cover problem asks whether $G$ has a vertex cover $S \subseteq V$ of size at most $\ell$. We construct a game as follows: the vertices controlled by $\Ev$, $V_\Ev=V$, correspond to the vertices of the original graph $G$ and the vertices controlled by $\Ad$, $V_\Ad=E$, correspond to the edges of the original graph $G$. 
From any vertex in $V_\Ev$, $\Ev$ can choose any edge $(u,v)$ of the original graph $G$ and move to the vertex corresponding $(u,v)\in V_\Ad$. From a vertex $(u,v)$ in $V_\Ad$, $\Ad$ can move to vertices $u,v\in V_\Ev$, corresponding to the two end points of the edge.
The start vertex is any arbitrary vertex in $V_\Ad$. 
The target set $\F=\{\{v\}~\mid~ v\in V_\Ev\}$ consists of singleton sets containing each vertex of $V$.
We claim that $G$ has a vertex cover of size at most $\ell$ iff $\Ad$ has a strategy to visit no more than $\ell$ targets.

Now suppose, $G$ has a vertex cover of size at most $\ell$, say $V'$. Then consider the strategy of $\Ad$ where he always picks a vertex from $V'$ whenever an edge is visited. It follows that under this strategy, no matter what $\Ev$ plays, only vertices from $V'$ are visited. Hence the number of targets visited is not more than $\ell$. 

On the other hand, suppose $\Ad$ has a strategy $\tau$ to enforce plays that visit no more than $\ell$ targets. Observe that since $\Ad$ is trying to minimise the number of targets visited, he has no incentive to choose different vertex for the same edge, when the edge is revisited. Hence, $\tau$ can be assumed to be a memoryless strategy. 
Now consider the following strategy $\sigma$ of $\Ev$: assuming any ordering on the set $E$, $\sigma$ chooses the next edge in this order in a round-robin manner. 
In the play under $\sigma$ and $\tau$, since every edge is picked by $\sigma$ at least once, and $\tau$ enforces no more than $\ell$ vertices are visited, the range of $\tau$ gives a vertex cover of size at most $\ell$. 
This proves our claim and completes the proof of \coNP-hardness.
\end{proof}

Note that the \coNP-hardness of the $\maxgenreach$ problem with target sets of size $1$ is in contrast with the results for $\genreach$ problem which is known to be in P for target sets of size $1$.
We leave open the question whether $\maxgenreach{}$ for singleton target sets is in \coNP. A possible approach would be to prove a bound on the memory of the winning strategy of $\Ad$, similar to the open case for size 2 target sets for $\genreach{}$. 

Next, we state our results for the two-player case of $\maxgenreachpromise$.

\begin{thm}
\label{thm:maxgrp2p}
\phantom{a}
\begin{enumerate}
\item The $\maxgenreachpromise$ problem for the two-player case is in $\P$ when all target sets are singleton.
\item $\maxgenreachpromise$ is \PSPACE-complete in general. It is $\PSPACE$-hard even when all target sets have size $3$.
\end{enumerate}
\end{thm}
\begin{proof}

We give a polynomial time algorithm for the case when all target sets are singleton.
The algorithm below is similar to the generalised reachability case and proceeds as follows:

\begin{enumerate}
    \item Compute the attractor relation between target vertices %
    forming a preorder graph capturing these relationships. Formally, with $t_0 = s$, consider the relation $t_i \preceq t_j$ iff $t_i \in \AtE(t_j)$, for $i\in \{0\}\cup [n]$ and create a graph with $t_i$ as vertices and edges from $t_i$ to $t_j$ iff $t_i \preceq t_j$. Call this the preorder graph.
    \item Perform the Strongly Connected Component (SCC) decomposition of this preorder graph.
    \item Assign a weight to each SCC equal to the number of target vertices contained within it.
    \item The resulting SCC decomposition is a Directed Acyclic Graph (DAG). Using bottom-up dynamic programming, find a path in this DAG with the maximum total %
    weight starting from the SCC containing $t_0$.
\end{enumerate}

The set of targets on this path from $t_0$ can be visited because, by construction, the $\Ev$-attractors of the targets are ordered by inclusion.

At the same time, any target set $\Ev$ can promise must have target states whose $\Ev$ attractors are ordered by inclusion; they are therefore on a path from $t_0$ in this DAG.
Thus, there cannot be a larger such set, as it would lead to a path with a larger weight.

This shows that our dynamic programming algorithm works. 
For the complexity, we note that each step can be performed efficiently:
     Attractor computation can be done in \( O(|V| + |E|) \). SCC decomposition can also be computed in \( O(|V| + |E|) \) \cite{SHARIR198167}.
     Assigning weights to SCCs and constructing the weighted DAG is linear in the graph size.
     Finally, computing the maximum-weight path in a DAG via bottom-up dynamic programming is also linear in the size of the DAG obtained via SCC decomposition. Since each step is polynomial, the overall algorithm solves the problem $\maxgenreachpromise$ in polynomial time when each target set is singleton.

\bigskip

For the general case with target sets of arbitrary size, we obtain a \NPSPACE~algorithm by simply guessing the target sets promised by $\Ev$ and then applying the \PSPACE~algorithm for solving generalised reachability games from Theorem~\ref{thm:pspace-hardness-general-case}. As a consequence of Savitch's Theorem, this gives a \PSPACE~upper bound for the general problem. The $\PSPACE$ lower bound holds even for target sets of size $3$ as the generalised reachability game with the objective to visit all targets is a special case of the $\maxgenreachpromise$ problem.
\end{proof}
We leave the exact complexity of $\maxgenreachpromise{}$ for the case with target sets of size $2$ open; Theorem \ref{thm:maxgrev} provides an \NP~lower bound.

\section{Discussion}

We study variations of games with generalised reachability objectives and maximum generalised reachability objectives, considering the size of target sets as parameter.
We provide several complexity results for these two problems, showing first
that in general, solving generalised reachability game is \PSPACE-complete.
Next, we show that generalised reachability can be checked in time linear in the size of the game, linear in the number of singleton target sets, and exponential in the number of larger target sets.
Further, we establish complexity results for the single-player cases, where either only the system or only the environment makes decisions.
Regarding sizes of winning strategies, we
provide tight upper and lower bounds for memory requirements of both players in generalised reachability games.

Next, we introduce optimisation variants of the generalised reachability problem, where the goal generalises from visiting all target sets to visiting as many target sets as possible.
The first natural goal here is to just \emph{maximise} the number of target sets visited; we show that this problem is different in that it is co\textsc{NP}-hard, even when each target set contains only a single element, and \textsc{PSPACE}-complete even when the size of target sets is restricted to at most $2$.
We also show that both single player variants of this problem are tractable if the target sets are singleton, but become intractable already for games with target sets of size $2$ over DAGs with pathwidth $2$.

We also consider an interesting variant, where $\Ev$ is asked to pledge -- before the game starts -- the target sets that are to be visited.
She then has to visit them all, and her goal shifts to pledging a largest set of target sets.
We show that the solution of games with such objectives is tractable for singleton target sets even in the two-player case. Another interesting variant could require $\Ev$ to specify the order in which the target sets are visited. This variant introduces additional constraints and may not be reducible to the cases considered in this work.

Thus, we clarify the landscape of complexity of several problems in generalised reachability. The major remaining open problem is the complexity of the generalised reachability problem where every target set is of size at most $2$. 
The promise variant, $\maxgenreachpromise$ problem with target sets of size $2$, is shown to be NP-hard (even when $\Ev$ is the only player). Since $\genreach$ can be seen as a special case of $\maxgenreachpromise$, this has consequences for the exact complexity of $\genreach$ with target sets of size $2$. In this case, a polynomial time algorithm for $\genreach$ would imply $\maxgenreachpromise$ is harder (unless $\P=\NP$). On the other hand, an upper bound better than PSPACE for $\maxgenreachpromise$ would also imply an improved upper bound for $\genreach$ in the case of target sets of size $2$.

Regarding the optimisation variant $\maxgenreach$, the exact complexity of the case with singleton sets remains open as we have shown it to be co\textsc{NP}-hard and in PSPACE.

\section*{Acknowledgments}

This work is supported by the EPSRC projects LaST (EP/Z003121/1), TRUSTED (EP/X03688X/1), and Games for Good (EP/X042596/1), the Fonds de la Recherche Scientifique – FNRS under Grant n° T.0188.23 (PDR ControlleRS) and by the EU's MSCA project Counting STaRS (grant agreement No 101285726).

%\nocite{}
\bibliographystyle{alphaurl} 
\bibliography{references}

@InProceedings{IgnatievJM13,
author="Ignatiev, Alexey
and Janota, Mikol{\'a}{\v{s}}
and Marques-Silva, Joao",
noeditor="J{\"a}rvisalo, Matti
and Van Gelder, Allen",
title="Quantified Maximum Satisfiability",
booktitle="Satisfiability Testing (SAT 2013)",
series       = {LNCS},
volume       = {7962},
doi          = {10.1007/978-3-642-39071-5\_19},
year="2013",
publisher="Springer",
noaddress="Berlin, Heidelberg",
pages="250--266",
isbn="978-3-642-39071-5"
}

@inproceedings{PR89,
author = {Pnueli, A. and Rosner, R.},
title = {On the synthesis of a reactive module},
year = {1989},
isbn = {0897912942},
publisher = {Association for Computing Machinery},
noaddress = {New York, NY, USA},
nourl = {https://doi.org/10.1145/75277.75293},
doi = {10.1145/75277.75293},
booktitle = {Symposium on Principles of Programming Languages (POPL 1989)},
pages = {179–190},
nonumpages = {12},
nolocation = {Austin, Texas, USA},
noseries = {POPL '89}
}

@misc{Zer13,
author = {Zermelo, E.},
title = {{Über eine Anwendung der Mengenlehre auf die Theorie des Schachspiels}},
year = {1913},
booktitle = {International Congress of Mathematicians (ICM 1913)}
}

@InProceedings{chatterjee_et_al10,
  author =	{Chatterjee, Krishnendu and Doyen, Laurent and Henzinger, Thomas A. and Raskin, Jean-Fran\c{c}ois},
  title =	{{Generalized Mean-payoff and Energy Games}},
  booktitle =	{Foundations of Software Technology and Theoretical Computer Science (FSTTCS 2010)},
  pages =	{505--516},
  series =	{LIPIcs},
  ISBN =	{978-3-939897-23-1},
  ISSN =	{1868-8969},
  year =	{2010},
  volume =	{8},
  noeditor =	{Lodaya, Kamal and Mahajan, Meena},
  publisher =	{Schloss Dagstuhl -- Leibniz-Zentrum f{\"u}r Informatik},
  noaddress =	{Dagstuhl, Germany},
  noURL =		{https://drops.dagstuhl.de/entities/document/10.4230/LIPIcs.FSTTCS.2010.505},
  URN =		{urn:nbn:de:0030-drops-28484},
  doi =		{10.4230/LIPIcs.FSTTCS.2010.505}
}

@article{GAREY1976237,
title = {Some simplified {NP}-complete graph problems},
journal = {Theoretical Computer Science},
volume = {1},
number = {3},
pages = {237-267},
year = {1976},
issn = {0304-3975},
doi = {10.1016/0304-3975(76)90059-1},
nourl = {https://www.sciencedirect.com/science/article/pii/0304397576900591},
author = {M.R. Garey and D.S. Johnson and L. Stockmeyer}
}

@article{ZIELONKA1998135,
title = {Infinite games on finitely coloured graphs with applications to automata on infinite trees},
journal = {Theoretical Computer Science},
volume = {200},
number = {1},
pages = {135-183},
year = {1998},
issn = {0304-3975},
doi = {10.1016/S0304-3975(98)00009-7},
nourl = {https://www.sciencedirect.com/science/article/pii/S0304397598000097},
author = {Wieslaw Zielonka}
}

@inproceedings{Karp72,
  author       = {Richard M. Karp},
  noeditor       = {Raymond E. Miller and
                  James W. Thatcher},
  title        = {Reducibility Among Combinatorial Problems},
  booktitle    = {Complexity of Computer Computations},
  series       = {The {IBM} Research Symposia Series},
  pages        = {85--103},
  publisher    = {Plenum Press, New York},
  year         = {1972},
  nourl          = {https://doi.org/10.1007/978-1-4684-2001-2\_9},
  doi          = {10.1007/978-1-4684-2001-2\_9},
}

@Inbook{Bloem2018,
author="Bloem, Roderick
and Chatterjee, Krishnendu
and Jobstmann, Barbara",
editor="Clarke, Edmund M.
and Henzinger, Thomas A.
and Veith, Helmut
and Bloem, Roderick",
title="Graph Games and Reactive Synthesis",
bookTitle="Handbook of Model Checking",
year="2018",
publisher="Springer International Publishing",
address="Cham",
pages="921--962",
isbn="978-3-319-10575-8",
doi="10.1007/978-3-319-10575-8_27",
nourl="https://doi.org/10.1007/978-3-319-10575-8_27"
}

@inproceedings{DBLP:journals/corr/abs-2509-14091,
  author       = {Sougata Bose and
                  Daniel Hausmann and
                  Soumyajit Paul and
                  Sven Schewe and
                  Tansholpan Zhanabekova},
  editor       = {Giorgio Bacci and
                  Adrian Francalanza},
  title        = {Generalised Reachability Games Revisited},
  booktitle    = {Games, Automata,
                  Logics, and Formal Verification (GandALF 2025)},
  series       = {{EPTCS}},
  pages        = {76--90},
  year         = {2025},
  url          = {https://doi.org/10.4204/EPTCS.428.7},
  doi          = {10.4204/EPTCS.428.7},
  bibsource    = {dblp computer science bibliography, https://dblp.org}
}

@InProceedings{chatterjee_et_al16,
  author =	{Chatterjee, Krishnendu and Dvor\'{a}k, Wolfgang and Henzinger, Monika and Loitzenbauer, Veronika},
  title =	{{Conditionally Optimal Algorithms for Generalized B\"{u}chi Games}},
  booktitle =	{Mathematical Foundations of Computer Science (MFCS 2016)},
  pages =	{25:1--25:15},
  series =	{LIPIcs},
  ISBN =	{978-3-95977-016-3},
  ISSN =	{1868-8969},
  year =	{2016},
  volume =	{58},
  noeditor =	{Faliszewski, Piotr and Muscholl, Anca and Niedermeier, Rolf},
  publisher =	{Schloss Dagstuhl -- Leibniz-Zentrum f{\"u}r Informatik},
  noaddress =	{Dagstuhl, Germany},
  noURL =		{https://drops.dagstuhl.de/entities/document/10.4230/LIPIcs.MFCS.2016.25},
  URN =		{urn:nbn:de:0030-drops-64403},
  doi =		{10.4230/LIPIcs.MFCS.2016.25}
}

@misc{fijalkow2023gamesgraphs,
      title={Games on Graphs: From Logic and Automata to Algorithms}, 
      author={Nathanaël Fijalkow and C. Aiswarya and Guy Avni and Nathalie Bertrand and Patricia Bouyer and Romain Brenguier and Arnaud Carayol and Antonio Casares and John Fearnley and Paul Gastin and Hugo Gimbert and Thomas A. Henzinger and Florian Horn and Rasmus Ibsen-Jensen and Nicolas Markey and Benjamin Monmege and Petr Novotný and Pierre Ohlmann and Mickael Randour and Ocan Sankur and Sylvain Schmitz and Olivier Serre and Mateusz Skomra and Nathalie Sznajder and Pierre Vandenhove},
      year={2025},
      doi={10.48550/arXiv.2305.10546},
}

@article{fij-horn2010,
author = {Fijalkow, Nathanaël and Horn, Florian},
year = {2010},
month = {10},
pages = {},
title = {The surprizing complexity of generalized reachability games},
doi = {10.48550/arXiv.1010.2420}
}

@article{Halaby2016,
  title={On the Computational Complexity of {MaxSAT}},
  author={Mohamed El Halaby},
  journal={Electron. Colloquium Comput. Complex.},
  year={2016},
  volume={TR16},
  url={https://api.semanticscholar.org/CorpusID:7872296}
}

@article{Kohli-et-al-94,
author = {Kohli, Rajeev and Krishnamurti, Ramesh and Mirchandani, Prakash},
title = {The Minimum Satisfiability Problem},
journal = {SIAM Journal on Discrete Mathematics},
volume = {7},
number = {2},
pages = {275-283},
year = {1994},
doi = {10.1137/S0895480191220836},

noURL = { 
    
        https://doi.org/10.1137/S0895480191220836
    
    

},
noeprint = { 
    
        https://doi.org/10.1137/S0895480191220836
    
    

}
}

@inproceedings{KupfermanS24,
  author       = {Orna Kupferman and
                  Noam Shenwald},
  noeditor       = {S. Akshay and
                  Aina Niemetz and
                  Sriram Sankaranarayanan},
  title        = {Games with Weighted Multiple Objectives},
  booktitle    = {Automated Technology for Verification and Analysis ({ATVA} 2024)},
  series       = {LNCS},
  volume       = {15054},
  pages        = {110--132},
  publisher    = {Springer},
  year         = {2024},
  nourl          = {https://doi.org/10.1007/978-3-031-78709-6\_6},
  doi          = {10.1007/978-3-031-78709-6\_6},
  bibsource    = {dblp computer science bibliography, https://dblp.org}
}

@article{DBLP:journals/jcss/Immerman81,
  author       = {Neil Immerman},
  title        = {Number of Quantifiers is Better Than Number of Tape Cells},
  journal      = {J. Comput. Syst. Sci.},
  volume       = {22},
  number       = {3},
  pages        = {384--406},
  year         = {1981},
  nourl          = {https://doi.org/10.1016/0022-0000(81)90039-8},
  doi          = {10.1016/0022-0000(81)90039-8},
  bibsource    = {dblp computer science bibliography, https://dblp.org}
}

@article{SHARIR198167,
title = {A strong-connectivity algorithm and its applications in data flow analysis},
journal = {Computers {\&} Mathematics with Applications},
volume = {7},
number = {1},
pages = {67-72},
year = {1981},
issn = {0898-1221},
doi = {10.1016/0898-1221(81)90008-0},
author = {M. Sharir}
}

@inproceedings{pnueli1977,
  author = {Pnueli, Amir},
  title = {The Temporal Logic of Programs},
  booktitle = {Foundations of Computer Science (FOCS 1977)},
  year = {1977},
  publisher    = {{IEEE} Computer Society},
  doi          = {10.1109/SFCS.1977.32},
  pages = {46--57}
}

@inproceedings{chatterjee2007,
  author = {Chatterjee, Krishnendu and Henzinger, Thomas A. and Piterman, Nir},
  title = {Generalized Parity Games},
  booktitle = {Foundations of Software Science and Computational Structures (FoSSaCS 2007)},
  series = {LNCS},
  publisher = {Springer},
  year = {2007},
  volume       = {4423},
  doi          = {10.1007/978-3-540-71389-0\_12},
  pages = {153--167}
}

@inproceedings{ChatterjeeHH11,
  author       = {Chatterjee, Krishnendu and Henzinger, Thomas A. and Horn, Florian},
  title        = {Finitary Winning in Omega-Regular Games},
  booktitle    = {ACM Transactions on Computational Logic},
  volume       = {11},
  number       = {1},
  pages        = {1--27},
  year         = {2011}
}

@article{FijalkowHorn13french,
  title = "Les jeux d'accessibilit{\'{e}} g{\'{e}}n{\'{e}}ralis{\'{e}}e",
  author = "Nathana{\"{e}}l Fijalkow and Florian Horn",
  year = "2013",
  journal = "Technique et Science Informatiques",
  number = "9-10",
  pages = "931--949",
  volume = "32",
  doi = "10.3166/tsi.32.931-949",
}

@inproceedings{ColcombetFH14,
  author       = {Thomas Colcombet and
                  Nathana{\"{e}}l Fijalkow and
                  Florian Horn},
  title        = {Playing Safe},
  booktitle    = {{FSTTCS}},
  series       = {LIPIcs},
  pages        = {379--390},
  publisher    = {Schloss Dagstuhl - Leibniz-Zentrum f{\"{u}}r Informatik},
  year         = {2014}
}
\end{document}